\documentclass[
    aps,
    prb,
    reprint,
    superscriptaddress,
    nofootinbib,
    longbibliography,
    floatfix
]{revtex4-2}

\usepackage{graphicx}
\usepackage{bm}
\usepackage{amsmath}
\usepackage{amssymb}
\usepackage{amsthm}
\usepackage{bbm}
\usepackage{mathtools}
\usepackage[stretch=50,shrink=50,step=1]{microtype}
\usepackage{xcolor}
\usepackage[colorlinks=true,linkcolor=blue,citecolor=blue,urlcolor=blue]{hyperref}

\newtheorem{lemma}{Lemma}

\newcommand{\ket}[1]{\left|#1\right\rangle}

\newcommand{\orcid}[1]{}

\begin{document}

\title{Exact solution of a boundary-driven transverse-field Ising model with hidden time-reversal symmetry}

\author{Xudong Liu\orcid{0000-0002-4605-7120}}
\email{liuxudong@iphy.ac.cn}
\affiliation{Beijing National Laboratory for Condensed Matter Physics, Institute of Physics, Chinese Academy of Sciences, Beijing 100190, China}
\affiliation{School of Physical Sciences, University of Chinese Academy of Sciences, Beijing 100049, China}

\author{Shu Chen\orcid{0000-0003-2605-6128}}
\email{schen@iphy.ac.cn}
\affiliation{Beijing National Laboratory for Condensed Matter Physics, Institute of Physics, Chinese Academy of Sciences, Beijing 100190, China}
\affiliation{School of Physical Sciences, University of Chinese Academy of Sciences, Beijing 100049, China}

\date{\today}

\begin{abstract}
The dissipative transverse-field Ising (TFI) model provides a paradigmatic setting for nonequilibrium quantum many-body physics. We show that a class of boundary-driven TFI models subject to dissipation at only one boundary possesses hidden time-reversal symmetry, which enables an exact construction of their nonequilibrium steady states. The solution admits a matrix-product representation and defines a nonequilibrium partition function from which steady-state observables can be evaluated efficiently. We use the exact solution to characterize the microscopic structure of the steady state
through its $z$--magnetization and two-point correlations. A striking feature is that the local field at the dissipative boundary governs the spatial organization of the steady state throughout the chain. The steady state typically exhibits boundary-localized magnetization profiles and exponentially decaying correlations, whose characteristic length scales are set by the dissipative boundary. In the weak-driving limit, suitably tuned boundary fields can reorganize the NESS into a delocalized single-interface structure, giving rise to long-range correlations that decay linearly with distance.
\end{abstract}

\keywords{transverse-field Ising model, dissipative systems, non-equilibrium steady state, Lindblad master equation, quantum detailed balance, hidden time-reversal symmetry, entangled state}

\maketitle

\section{Introduction}
Open quantum systems have attracted broad interest in recent years.
From a fundamental perspective, realistic quantum systems are inevitably
coupled to their environments, and the interplay between coherent Hamiltonian
evolution and quantum dissipation gives rise to rich nonequilibrium phenomena~\cite{Diehl2008,Garrahan2010,Diehl2016,Rossini2021,Poletti_RMP}.
From a practical perspective, engineered dissipation makes it possible to
control quantum dynamics and prepare entangled quantum states~\cite{Zoller1996,WhiteNoise,entangle2013,Clerk_PRB2022,Clerk_PRX2024,CQA,Zoller}, with applications in quantum information processing~\cite{QI_1,QI_2,QI_3,QI_4}.
Despite these distinct motivations, the nonequilibrium steady state (NESS)
is a central object of interest.
However, since the dimension of the Hilbert space increases exponentially with the system size, determining the steady states of open quantum systems is generally a formidable numerical task.
Exact steady states are therefore particularly valuable.

Even for Markovian open quantum systems described by Lindblad master equations~\cite{Lindblad,book_1,lidar}, exact solutions remain rare.
A major class of exactly tractable models relies on a closed hierarchy of
equations of motion for correlation functions. Typical examples include quasi-free fermionic and bosonic systems which can be solved by third quantization~\cite{prosen2008,prosen2010JPA,prosen2010JSM,poletti2018,poletti2018_2}, as well as certain quartic Lindbladian systems subject to specific constraints~\cite{clerk_weaksymm,ZYK,Wang}.
Another class of exact solutions relies on quantum integrability. The Yang--Baxter structure contains an additional degree of freedom that
gives rise to nonunitary representations of quantum groups.
These representations make it possible to construct exact steady states in matrix-product-operator (MPO) form for a set of boundary-driven integrable models~\cite{prosen2011xxz,prosen_mps,popkov2013xxz,
prosen2014hubbard,Ilievski2017,Ilievski2014,prosen2020xyz,prosen2020xyz2,prosen2022xyz,prosen2025xxz}.
These exact steady states have revealed new conserved quantities and played a central role in understanding nonequilibrium transport and dissipative phase transitions~\cite{prosen2013quasilocal,prosen2011drude,clerk_xxz}.

A conceptually distinct route to exact NESS is provided by hidden time-reversal symmetry (HTRS)~\cite{clerk_prxq}. This symmetry is closely related to the Kubo--Martin--Schwinger (KMS) detailed balance~\cite{CarlenMaas2017,GQDB,SQDB}. In contrast to the well-known Gelfand-Naimark-Segal (GNS) quantum detailed balance~\cite{Agarwal1973,Alicki1976,PhysicsKnight}, systems with HTRS do not require Onsager time symmetry to hold for all correlators. Consequently, they are allowed to have nontrivial steady states that do not commute with the Hamiltonian. Moreover, HTRS need not rely on the rotating-wave approximation~\cite{Alhambra}, extending its applicability to a broader class of quantum many body systems. 
Beyond the formal generalization, HTRS imposes powerful constraints on the Lindbladian that can be exploited to find exact solutions. This approach
may be viewed as a generalization of the coherent quantum absorber (CQA)
method~\cite{CQA} and has already yielded exact steady states in several
dissipative quantum systems~\cite{clerk_kerr,clerk_ising,clerk_bose,clerk_xxz,clerk_fermion}.

In this work, we present an exact solution for a boundary-driven TFI model with two-sided longitudinal fields and single-sided dissipation. In contrast to the commonly solvable dissipative TFI models~\cite{prosen2008,ZhengZY,Naoyuki}, the equations of motion for
$k$-point correlation functions do not form a closed hierarchy, rendering
third quantization inapplicable. Nevertheless, we show that the model
possesses HTRS, which allows us to construct an exactly solvable doubled
system consisting of two TFI chains coupled through a unidirectional
waveguide [see Fig.~\ref{fig1}(a)]. The NESS of the original system is then
obtained from the pure steady state of the doubled system.

We further derive an MPO representation of the NESS and formulate a
nonequilibrium partition function, which enables physical observables to be
evaluated efficiently using the transfer operator of an effective classical
Markov process. We use this framework to investigate the steady-state
$z$-magnetization and connected two-point correlations under different
boundary-field configurations. We find that the spatial structure of the
NESS depends sensitively on the boundary fields, with the strength of the
longitudinal field at the dissipative boundary playing a particularly
important role throughout the chain. For generic parameters, the
magnetization is exponentially localized near the boundaries and the
connected correlations decay exponentially with distance. Both the
magnetization localization length and the correlation decay rate are
governed by the strength of the dissipative boundary field. In the weak-driving
limit, as this field strength approaches the Ising interaction strength,
the magnetization localization length diverges, accompanied by a pronounced
change in the correlation. 
Moreover, for suitable combinations of the two boundary-field strengths,
the purification of the NESS becomes dominated by configurations containing
a single delocalized interface separating two regions, resulting in a linear magnetization profile and long-range correlations that decay linearly with distance.

The remainder of this paper is organized as follows. In Sec.~\ref{Sec2}, we introduce the boundary-driven TFI model and the basic concepts of HTRS. In Sec.~\ref{Sec3}, we show how this symmetry enables the exact solution of the NESS and then derive its matrix-product representation. In Sec.~\ref{Sec4}, we investigate how the field at the dissipative boundary controls the spatial structure of the NESS through the magnetization profiles and connected two-point correlations. 
In Sec.~\ref{Sec5}, we give a summary.
Technical derivations and proofs of the main results are provided in the Appendices.

\section{Models and methods}\label{Sec2}
\subsection{Boundary-driven TFI model}
We consider a 1-D chain of $N$ qubits with nearest-neighbor $ZZ$ interactions and a transverse field, subject to longitudinal boundary fields and gain at site $N$. The dynamics of this system are governed by the Lindblad master equation
\begin{equation}
\frac{d\rho}{dt}=\mathcal{L}\rho=-i[H,\rho]+\mathcal{D}[\sqrt{\gamma}\sigma^{+}_{N}]\rho,
\label{model}
\end{equation}
where $H=H_{\text{TFI}}+H_l$ comprises the transverse-field Ising part and the boundary terms:
\begin{align}
&H_{\text{TFI}}=\sum^{N-1}_{j=1}\sigma^{z}_{j}\sigma^{z}_{j+1}+h\sum^{N}_{j=1}\sigma^{x}_{j},\\
&H_l=h_L\sigma^{z}_{1}+h_R\sigma^{z}_{N}.
\end{align}
The dissipator is given by $\mathcal{D}[L]\rho=2L\rho L^{\dagger}-\left\{L^{\dagger}L,\rho \right\}$. These ingredients are compatible with programmable atomic platforms, such as Rydberg-atom arrays. 

It should be noted that Eq.~\eqref{model} differs from the boundary driven TFI models studied in Refs.~\cite{prosen2008,ZhengZY,Naoyuki}, which can be solved via third quantization. Under the spin rotation
$
U=\prod_{j=1}^{N} e^{-i\frac{\pi}{4}\sigma_j^{y}},
$
the jump operators $\sigma^+$ are transformed into $\frac{1}{2}(\sigma^z + i\sigma^y)$. After the Jordan--Wigner transformation to the Majorana fermion representation, it becomes clear that the resulting jump operators are neither purely linear nor purely quadratic, but rather a combination of both. Consequently, the presence of terms such as $\sigma^{y}\rho\sigma^{z}$ in the Lindblad equation breaks the closure of correlation functions, making the solution difficult to obtain.

\subsection{Hidden time-reversal symmetry}
In this work, we focus on the NESS of Eq.~\eqref{model}, i.e., density matrices $\rho_{\rm{SS}}$ satisfying $\mathcal{L}\rho_{\rm{SS}}=0$. 
In Appendix~\ref{appA}, we proved that the NESS is unique and full rank.

A route to exact steady states is provided by the KMS quantum detailed balance condition~\cite{clerk_prxq,GQDB,SQDB}, which can be formulated in terms of a hidden time-reversal symmetry (HTRS) in the doubled Hilbert space~\cite{clerk_prxq}.
To formulate the HTRS, we introduce the mirrored copy $\mathcal{H}_B$ of the physical Hilbert space $\mathcal{H}_A$. Given a time-reversal operator $T$, a purification of the NESS can be written as
\begin{equation}
    \rho_{\rm SS}
    =\operatorname{Tr}_{B}\!\left(|\Psi_T\rangle\langle\Psi_T|\right),
    \qquad
    |\Psi_T\rangle
    =\sum_n\sqrt{p_n}\,|n\rangle_A T|n\rangle_B,
\label{ness_psiT}
\end{equation}
where $p_n$ and $|n\rangle$ are the eigenvalues and eigenvectors of
$\rho_{\rm SS}$, respectively.
The HTRS holds if $T$ can be chosen such that
\begin{equation}
    \langle\Psi_T|X_A(t)Y_B|\Psi_T\rangle
    =
    \langle\Psi_T|Y_A(t)X_B|\Psi_T\rangle
\label{hTRS}
\end{equation}
for arbitrary operators $X$ and $Y$.
Here, $X_A=X\otimes\mathbb{I}_B$ and $X_B=\mathbb{I}_A\otimes X$, with the time evolution generated by $\mathcal{L}\otimes\mathbb{I}_B$.
An equivalent formulation, which does not rely on the auxiliary system, reads
\begin{equation}
\mathrm{Tr}\!\left(X(t)\,\mathcal{J}[Y]\,\rho_{\mathrm{SS}}\right)
=
\mathrm{Tr}\!\left(Y(t)\,\mathcal{J}[X]\,\rho_{\mathrm{SS}}\right),
\label{hTRS2}
\end{equation}
for arbitrary observables $X$ and $Y$. Here, the exchange superoperator is defined as
$
\mathcal{J}[O]
=
\rho_{\mathrm{SS}}^{1/2}\, T O^\dagger T^{-1}\, \rho_{\mathrm{SS}}^{-1/2},
$
which satisfies
$
\mathcal{J}[O]_B|\Psi_T\rangle = O_A|\Psi_T\rangle.
$
Compared with the well-known GNS quantum detailed balance condition, $\mathrm{Tr}\!\left(X^{\dagger}(t)Y\rho_{\mathrm{SS}}\right)=\mathrm{Tr}\!\left(Y^{\dagger}(t)X\rho_{\mathrm{SS}}\right)$, Eq.~\eqref{hTRS2} is less restrictive and therefore admits nontrivial NESS satisfying $[\rho_{\mathrm{SS}},H]\neq0$. Moreover, for $\mathcal{J}$-invariant operators, such as the effective Hamiltonian $H_{\mathrm{eff}}$ and the jump operator $L$, Eq.~\eqref{hTRS2} reduces to the standard Onsager symmetry,
\begin{equation}
\mathrm{Tr}\!\left(X(t)Y\rho_{\mathrm{SS}}\right)=\mathrm{Tr}\!\left(Y(t)X\rho_{\mathrm{SS}}\right).
\label{standard Onsager symmetry}
\end{equation}
Fig.~\ref{fig1}(b) demonstrates this behavior explicitly for Lindbladian~\eqref{model}: $\mathcal{J}$-invariant operators obey the symmetry in Eq.~\eqref{standard Onsager symmetry}, while generic Hermitian operators generally do not. This provides a clear signature of the HTRS and demonstrates that Lindbladian~\eqref{model} does not satisfy the conventional GNS detailed balance condition.

Since the explicit form of $T$ is generally unknown, Eq.~\eqref{hTRS} is difficult to apply directly to solve NESS. Instead, it is more convenient to use a sufficient condition for HTRS. For Lindbladian~\eqref{model}, this condition can be shown to take the form \cite{clerk_prxq}
\begin{align}
    &H_{AB}|\Psi_{T}\rangle=0,\label{hTRS_1}\\
    &(\sigma^{+}_{B,N}-\sigma^{+}_{A,N})|\Psi_{T}\rangle=0,\label{hTRS_right}
\end{align}
where 
\begin{equation}
H_{AB}=H_A-H_B-i\gamma\left(\sigma^{-}_{A,N}\sigma^{+}_{B,N}-\text{h.c.}\right).
\label{H_{AB}}
\end{equation}
Here, $H_A$ and $H_B$ represent the copies of Hamiltonian $H$ acting on the respective subspaces, and $\sigma^{+}_{A,N}$ and $\sigma^{+}_{B,N}$ are similar.
In other words, if the doubled system 
\begin{equation}
    \mathcal{L}_{AB}=-i[H_{AB},\cdot]+\gamma\mathcal{D}[\sigma^{+}_{N,A}-\sigma^{+}_{N,B}]
\label{doubled_system}
\end{equation}
has a pure state as its unique NESS \cite{Zoller}, this pure state is the purification of the NESS of the physical subsystem. Eq.~\eqref{doubled_system} describes the dissipation in system $A$ absorbed by the mirror system $B$ via a chiral waveguide [see Fig.~\ref{fig1}(a)].

\begin{figure}[t]
 \centering
        \includegraphics[width=\linewidth]{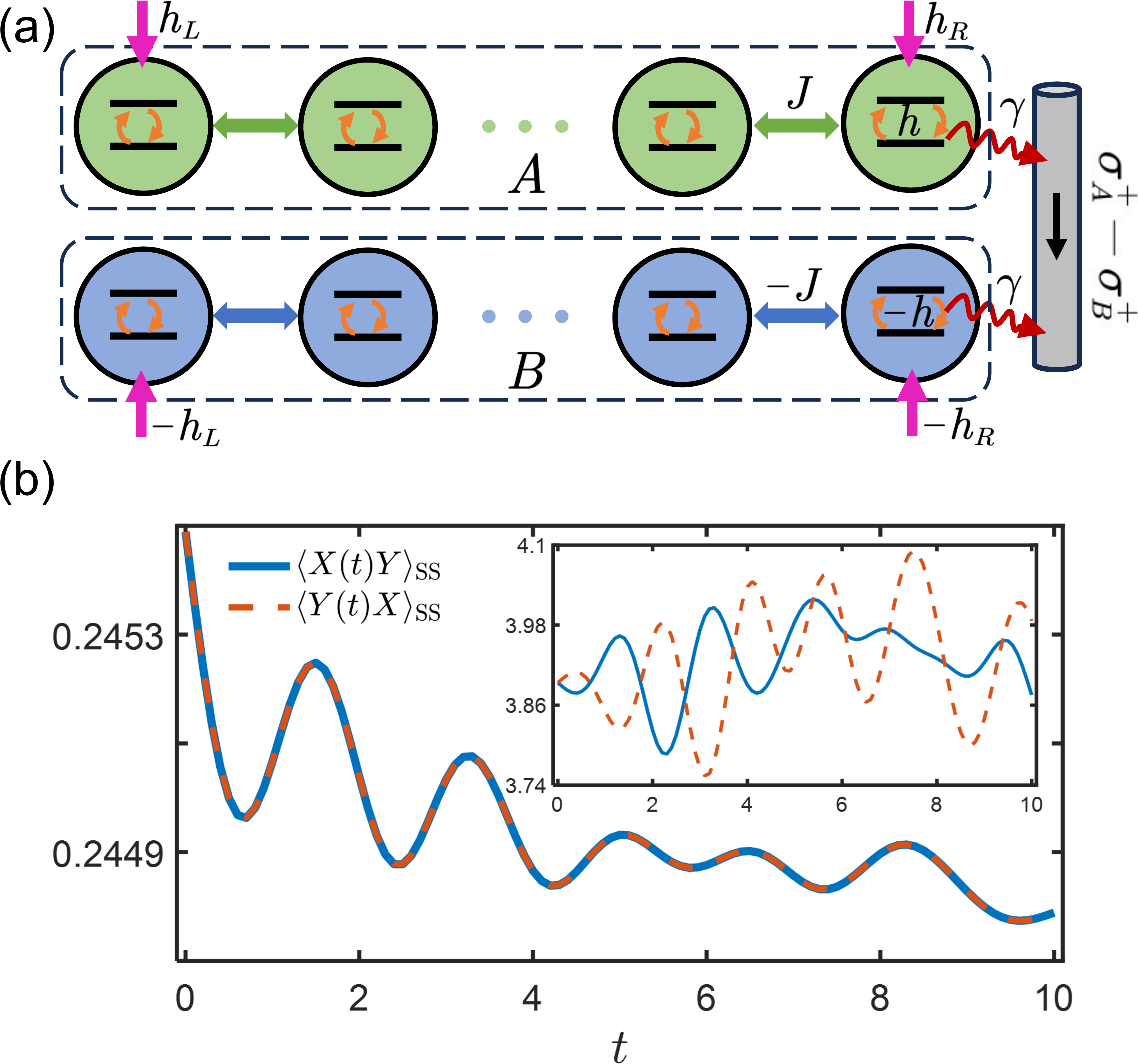}
 \caption{(a) Schematic of the doubled system for the boundary-driven TFI model. The physical system $A$ consists of $N$ spins with longitudinal boundary fields and incoherent gain at site $N$ (right). The mirror system $B$ acts as an absorber and is coupled to $A$ via a unidirectional waveguide.
 (b) Features of HTRS: two time correlators for a small chain $N=5$ with parameters $h=0.5$, $h_L=h_R=0.5$ and $\gamma=0.1$. It shows a symmetric correlation between $X=H_{\rm{eff}}$ and $Y=\sqrt{\gamma}\sigma^{+}_{N}$. This symmetry does not hold for general operators (inset: $X=\sigma^{z}_{1}$ and $Y=\sigma^{y}_{2}+5\sigma^{z}_{4}$).}
\label{fig1}
\end{figure}

\section{Exact nonequilibrium steady state}
\label{Sec3}

In this section, we show how to construct the NESS satisfying Eqs.~\eqref{hTRS_1}--\eqref{hTRS_right}. Inspired by the exact NESS solutions of boundary driven XXZ models \cite{prosen_mps, clerk_xxz}, we propose the following ansatz:
\begin{equation}
    |\Psi_T\rangle={\underset{S_i\in\{0,1,t\}}{\sum}}
    \text{Coeff}(S_1S_2\cdots S_{N-1}S_N)
    |S_1S_2\cdots S_{N-1} S_{N}\rangle,
\label{MPA}
\end{equation}
where $\text{Coeff}(S_1S_2\cdots S_N)$ denotes the coefficient associated with a given configuration of local states, to be determined by the HTRS constraints. $S_i$ denotes the coupling state of subsystems A and B at the site $i$. The local basis $|1\rangle, |0\rangle, |s\rangle, |t\rangle$ are defined as 
$|\uparrow_A\uparrow_B\rangle$, $|\downarrow_A\downarrow_B\rangle$, $\frac{1}{\sqrt{2}}(|\uparrow_A\downarrow_B\rangle-|\downarrow_A\uparrow_B\rangle)$, $\frac{1}{\sqrt{2}}(|\uparrow_A\downarrow_B\rangle+|\downarrow_A\uparrow_B\rangle)$ 
respectively. 
% In matrix language, they correspond to $\left\{\sigma^+,\sigma^-,\sigma^z/\sqrt{2},\sigma^0/\sqrt{2} \right\}$ with $\sigma^0$ denoting the identity. 

It should be noted that, in the ansatz \eqref{MPA}, the local states $|S_i\rangle$ are restricted to $\{|0\rangle,|1\rangle,|t\rangle\}$ and do not include the spin singlet state $|s\rangle$. Acting with $H_{AB}$ on any product state $|S_1 S_2 \cdots S_N\rangle$ ($S_i \in \{0,1,t\}$) generates a linear combination of product states in which at most one site is in $|s\rangle$. These generated $|s\rangle$ play a crucial role in solving Eq.~\eqref{hTRS_1}, as will be shown below. This construction is closely related to the so-called ``isolated defect operator" (IDO) approach \cite{prosen_mps}.

We first analyze how Eq.~\eqref{hTRS_right} constrains the steady state. For the $N$-th site, it is straightforward to obtain:
\begin{equation}
    (\sigma^{+}_{B}-\sigma^{+}_{A})|0\rangle=-\sqrt{2}|s\rangle,\;
    (\sigma^{+}_{B}-\sigma^{+}_{A})|1\rangle=
    (\sigma^{+}_{B}-\sigma^{+}_{A})|t\rangle=0.
\end{equation}
Since our ansatz~\eqref{MPA} does not include $|s\rangle$, Eq.~\eqref{hTRS_right} requires that the $N$-th site cannot be in $|0\rangle$. Physically speaking, this result is due to the requirement of dissipation for the polarization of the boundary spin.

To find $\text{Coeff}(S_1S_2\cdots S_N)$ in the ansatz~\eqref{MPA}, we then examine the action of $H_{AB}$ on $|S_1\cdots S_N\rangle$. Note that $H_{AB}$ can be divided into two-body terms, transverse-field terms, and boundary driven terms:
\begin{equation}
    \hat{H}_{AB}=\sum^{N-1}_{j=1} \hat{h}^{j,j+1}_{\rm{Ising}} +
                \sum^{N}_{j=1}\hat{h}^{j}_{\rm{tf}} + \hat{h}_{\rm{left}} + \hat{h}_{\rm{right}},
\end{equation}
where $\hat{h}^{j,j+1}_{\rm{Ising}}=\sigma^{z}_{A,j}\sigma^{z}_{A,j+1}-\sigma^{z}_{B,j}\sigma^{z}_{B,j+1}$, $\hat{h}^{j}_{\rm{tf}}=h\left(\sigma^{x}_{A,j}-\sigma^{x}_{B,j}\right)$, $\hat{h}_{\rm{left}}=h_L\left(\sigma^{z}_{A,1}-\sigma^{z}_{B,1}\right)$, and $\hat{h}_{\rm{right}}=h_R\left(\sigma^{z}_{N,A}-\sigma^{z}_{N,B}\right)-i\gamma\left(\sigma^{-}_{A,N}\sigma^{+}_{B,N}-\text{h.c.}\right)$.The non-vanishing results obtained by applying these local terms to the local basis are as follows:
\begin{equation}
\begin{aligned}
    &\hat{h}^{j,j+1}_{\text{Ising}}\ket{1_jt_{j+1}}=2\ket{1_js_{j+1}},\;
    \hat{h}^{j,j+1}_{\text{Ising}}\ket{t_j1_{j+1}}=2\ket{s_j1_{j+1}},
    \\
    &\hat{h}^{j,j+1}_{\text{Ising}}\ket{0_jt_{j+1}}=-2\ket{0_js_{j+1}},\quad
    \hat{h}^{j}_{\rm{tf}}\ket{0_j}=\sqrt{2}h\ket{s_j},\\
    &\hat{h}^{j,j+1}_{\text{Ising}}\ket{t_j0_{j+1}}=-2\ket{s_j0_{j+1}},\quad
    \;\hat{h}^{j}_{\rm tf}\ket{1_j}=-\sqrt{2}h\ket{s_j},\\
    &\hat{h}_{\rm left}\ket{t_1}=2h_L\ket{s_1},\qquad
    \hat{h}_{\rm right}\ket{t_N}=(i\gamma+2h_R)\ket{s_N}.
\label{localterm}
\end{aligned}
\end{equation}
It is straightforward to see that acting with $H_{AB}$ on any allowed configuration 
$|S_1S_2\cdots S_N\rangle$ in the ansatz~\eqref{MPA} generates a superposition of states 
with a single local singlet, 
$|S_1\cdots S_{j-1}s_j S_{j+1}\cdots S_N\rangle$. 
Consequently, Eq.~\eqref{hTRS_1} is equivalent to requiring that the sum of the coefficients in front of every single-$s$ state vanishes in $H_{AB}|\Psi_T\rangle$. 
Substituting Eqs.~\eqref{localterm} into the left hand side of Eq.~\eqref{hTRS_1}, the vanishing of singlets in the bulk yields the following coefficient relations:

\begin{equation}
\begin{aligned}
    &\text{Coeff}(\dots\alpha_{j-1}0_j \beta_{j+1}\dots )-\;\text{Coeff}(\dots \alpha_{j-1}1_j \beta_{j+1}\dots)\\
    &=\quad\frac{\zeta^{\alpha}_{\beta}}{\sqrt{2}h}\;\text{Coeff}(\dots\alpha_{j-1}t_j \beta_{j+1}\dots )
\end{aligned}
\label{specific_constraints_0}
\end{equation}
Here, $\dots$ denotes that in each term of the equation all sites except $j-1$, $j$ and $j+1$ have identical states. The indices $\alpha,\beta\in\{0,1,t\}$ label the neighboring local states, and $\zeta^\alpha_\beta$ are structure constants with the only nonzero elements
\begin{equation}
    \zeta^{0}_{t}=-\zeta^{1}_{t}=2,\quad \zeta^{0}_{0}=-\zeta^{1}_{1}=4,\quad\zeta^{\alpha}_{\beta}=\zeta^{\beta}_{\alpha}.
\label{parameter_bulk}
\end{equation}
The range of $j$ depends on the neighboring states appearing in Eq.~\eqref{specific_constraints_0}. 
For terms with $\beta=0$, $2\le j\le N-2$ due to the boundary constraint imposed by the right dissipator at site $N$, while for $\beta=1$ the allowed range is $2\le j\le N-1$.
Similarly, the vanishing of singlets at the boundary gives additional relations among the coefficients:
\begin{equation}
\begin{aligned}
    & \text{Coeff}(1_{1}\alpha_{2}\dots)-\text{Coeff}(0_{1}\alpha_{2}\dots)
    =\frac{\zeta^{\alpha}_{L}}{\sqrt{2}h}\text{Coeff}(t_{1}\alpha_{2}\dots),\\
    & \text{Coeff}(\dots\alpha_{N-1}1_{N})=\frac{\zeta^{\alpha}_{R}}{\sqrt{2}h}\text{Coeff}(\dots\alpha_{N-1}t_{N}),\label{specific_constraints_boundary}
\end{aligned}
\end{equation}
where $\alpha\in\{0,1,t\}$ and the boundary structure constants are
\begin{equation}
\begin{aligned}
&\zeta_L^{0}=2h_L-2,\;
\zeta_L^{1}=2h_L+2,\;
\zeta_R^{t}=i\gamma+2h_R\\
&\zeta_L^{t}=2h_L,\;
\zeta_R^{0}=i\gamma+2h_R-2,\;
\zeta_R^{1}=i\gamma+2h_R+2.
\label{parameter_boundary}
\end{aligned}
\end{equation}

Eqs.~\eqref{specific_constraints_0}-\eqref{parameter_boundary} determine the coefficients in the ansatz~\eqref{MPA} up to an overall normalization. 
For convenience, we choose the normalization convention
$
\text{Coeff}(tt\cdots tt)=1 .
$
To state the solution compactly, we introduce the following two definitions.

\textit{Definition} 1 (domain).
For a configuration $\ket{S_1S_2\cdots S_N}$ with $S_j\in\{0,1,t\}$, we say that it has a length-$l$ domain on sites $k,k+1,\ldots,k+l-1$ if
\[
S_j\neq t,\qquad j=k,k+1,\ldots,k+l-1,
\]
and the neighboring sites, if present, are in $\ket{t}$, i.e.,
\[
S_{k-1}=t\quad (k>1),\qquad S_{k+l}=t\quad (k+l\le N).
\]
For a domain with $l<N$, we call it a left-boundary domain if it contains site $1$, a right-boundary domain if it contains site $N$, and a bulk domain otherwise. In particular, we say $\ket{tt\cdots t}$ has a length-zero domain.

\textit{Definition} 2 (bond).
For a configuration $\ket{S_1S_2\cdots S_N}$ with $S_j\in\{0,1,t\}$, we say that it has a \textit{type-I bond} between sites $j$ and $j+1$ if $(S_j,S_{j+1})=(0,0)$ or $(1,1)$, and a \textit{type-II bond} between these sites if $(S_j,S_{j+1})=(0,1)$ or $(1,0)$.

For example, the configuration $\ket{1t110t11}$ has a length-1 left-boundary domain, a length-3 bulk domain, and a length-2 right-boundary domain. It contains two type-I bonds and one type-II bond in total. Fig.~\ref{fig2} (a) provides an intuitive illustration of the domain and bond structure.

With the two definitions above, Eqs.~\eqref{specific_constraints_0}-\eqref{parameter_boundary} can be solved in closed form.
The resulting coefficient takes a domain-factorized form: the weight of
each configuration in the ansatz~\eqref{MPA} is fixed solely by its domain
structure and bond content (see Appendix~\ref{appB} for the proof). Explicitly,
\begin{equation}
\mathrm{Coeff}(S_1S_2\cdots S_N)
=
\lambda_{\rm I}^{N_{\rm I}}
\lambda_{\rm II}^{N_{\rm II}}
\prod_j c_j ,
\label{coeff}
\end{equation}
where $j$ labels the domains in the configuration
$\ket{S_1S_2\cdots S_N}$, and $N_{\rm I}$ and $N_{\rm II}$ denote the
total numbers of type-I and type-II bonds, respectively. The corresponding
bond weights are
\begin{equation}
\lambda_{\rm I}
=
\frac{i\gamma+2h_R+2}{2\sqrt{2}h},
\qquad
\lambda_{\rm II}
=
\frac{i\gamma+2h_R-2}{2\sqrt{2}h}.
\label{bond_weight}
\end{equation}
The domain type factor $c_j$ depends only on the type of the domain. It is useful
to introduce
\begin{equation}
c_{\rm b}
=
\frac{i\gamma+2h_R}{2\sqrt{2}h},
\qquad
c_{\pm}
=
\frac{i\gamma+2h_R\pm2h_L}{2\sqrt{2}h}.
\label{domain_factors}
\end{equation}
Then $c_j=c_{\rm b}$ for a bulk domain, $c_j=1$ for a length-zero domain,
$c_j=c_+$ ($c_-$) for a left-boundary domain beginning with $\ket{1}$ ($\ket{0}$),
and $c_j=2c_{\rm b}$ for a right-boundary domain.
A length-$N$ domain carries twice the factor of the corresponding left-boundary domain.
Eqs.~\eqref{ness_psiT}, \eqref{MPA}, and \eqref{coeff} together provide an exact solution of the NESS. One can readily verify that this NESS is genuinely nonequilibrium, satisfying $[\rho_{\text{SS}},H] \neq 0$.

The exact NESS admits an MPO representation, which provides a convenient
framework for the numerical evaluation of physical observables.
As shown in Eq.~\eqref{coeff}, the NESS is fully characterized by
the bond structure and the domain-type factors.
Notice that the domain type factors only depend on the local environments
at the two domain boundaries, and therefore can be regarded as contributions
from additional boundary bonds. Fig.~\ref{fig2}(a) gives a schematic illustration of the above bonds.
In this way, the coefficient formula becomes fully local in bond representation, leading directly to the following MPS representation of $|\Psi_T\rangle$:
\begin{equation}
|\Psi_T\rangle
=
\langle \psi_L|
K_L L^{[2]}K L^{[3]}K\cdots
L^{[N-1]}K_R
|\psi_R\rangle .
\label{MPO_psiT}
\end{equation}
Here multiplication is performed in the auxiliary space. The
bulk tensors are
\begin{align}
L^{[j]}
&=
\begin{bmatrix}
\ket{1_j} & 0 & 0\\
0 & \ket{t_j} & 0\\
0 & 0 & \ket{0_j}
\end{bmatrix},
&
K
&=
\begin{bmatrix}
\lambda_{\rm I} & \sqrt{c_{\rm b}} & \lambda_{\rm II}\\
\sqrt{c_{\rm b}} & 1 & \sqrt{c_{\rm b}}\\
\lambda_{\rm II} & \sqrt{c_{\rm b}} & \lambda_{\rm I}
\end{bmatrix}.
\label{mpo_defs}
\end{align}
The boundary tensors are $\langle\psi_L|
=\bigl[\ket{1_1}\;\ket{t_1}\;\ket{0_1}\bigr]$, $|\psi_R\rangle=\bigl[\ket{1_N}\;\ket{t_N}\bigr]^T$,
\begin{align}
&K_L
=
\begin{bmatrix}
c_{+}\lambda_{\rm I}/\sqrt{c_{\rm b}} 
& c_{+}
& c_{+}\lambda_{\rm II}/\sqrt{c_{\rm b}}
\\
\sqrt{c_{\rm b}}
& 1
& \sqrt{c_{\rm b}}
\\
c_{-}\lambda_{\rm II}/\sqrt{c_{\rm b}}
& c_{-}
& c_{-}\lambda_{\rm I}/\sqrt{c_{\rm b}}
\end{bmatrix},\\
&
K_R
=
\begin{bmatrix}
2\lambda_{\rm I}\sqrt{c_{\rm b}} & \sqrt{c_{\rm b}}\\
2c_{\rm b} & 1\\
2\lambda_{\rm II}\sqrt{c_{\rm b}} & \sqrt{c_{\rm b}}
\end{bmatrix}.
\label{mpo_boundary}
\end{align}
For numerical implementation, it is useful to absorb the bond tensors $K$, $K_L$, and $K_R$ into the local tensors, yielding
\begin{equation}
    |\Psi_T\rangle=T^{[1]}T^{[2]}\cdots T^{[N]},
\label{MPO_psiT_Tmatrix}
\end{equation}
with $T^{[1]}=\langle\psi_L|$, $T^{[N]}=|\psi_R\rangle$,
$T^{[2]}=K_LL^{[2]}\sqrt K$, $T^{[N-1]}=\sqrt K L^{[N-1]}K_R$, and
$T^{[j]}=\sqrt K L^{[j]}\sqrt K$ for $3\le j\le N-2$.
This MPS is illustrated schematically in Fig.~\ref{fig2} (b).
By replacing the local states $|0\rangle$, $|1\rangle$, and $|t\rangle$
in $|\Psi_T\rangle$ by $\sigma_A^{-}$, $\sigma_A^{+}$, and
$\sigma_A^{0}/\sqrt{2}$, respectively, we can map the MPS representation of
$|\Psi_T\rangle$ onto an MPO representation of an operator
$\Omega\in\mathrm{End}(\mathcal H_A)$. The NESS is then given by
$\rho_{\mathrm{SS}}=\Omega\Omega^\dagger/\mathrm{Tr}(\Omega\Omega^\dagger)$.

\begin{figure}[t]
 \centering
        \includegraphics[width=0.48\textwidth]{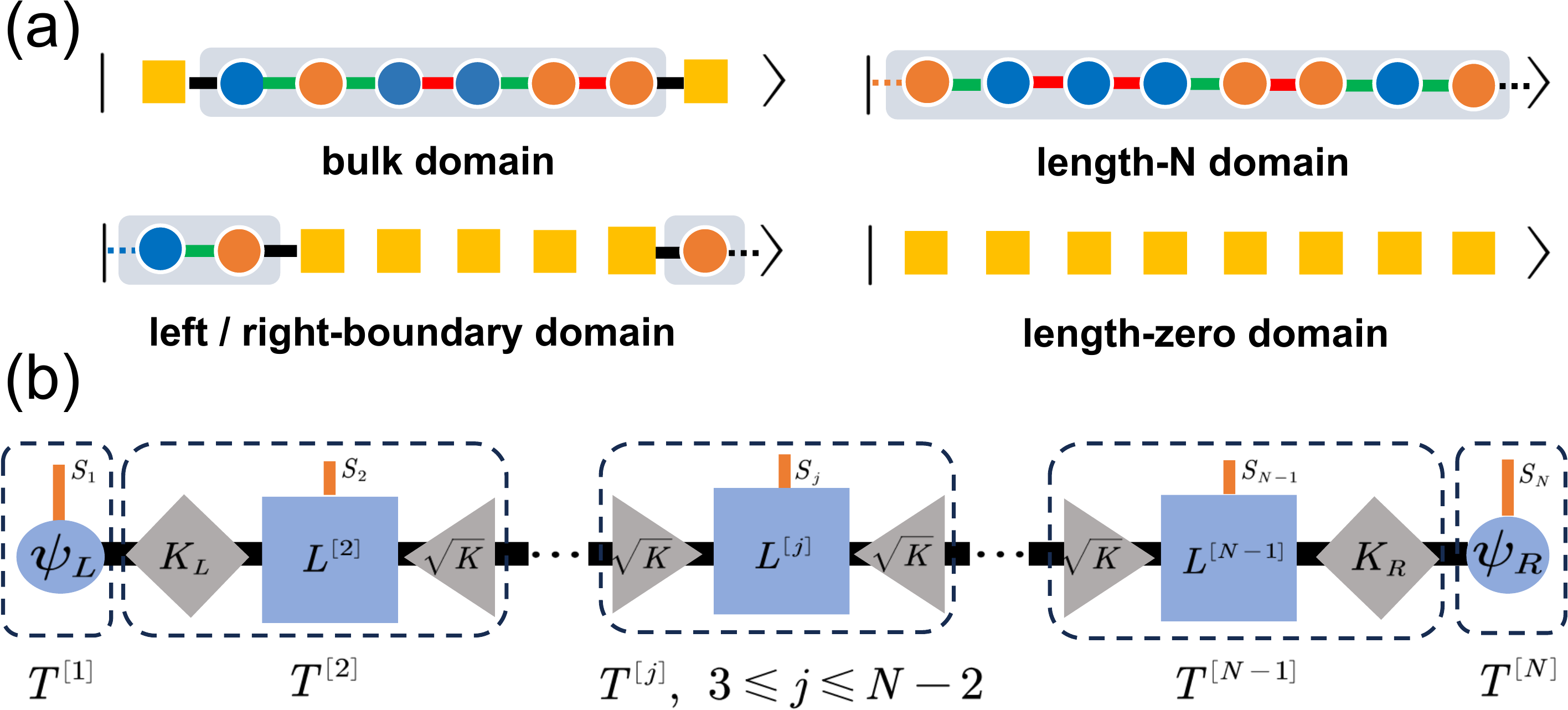}
 \caption{(a) Schematic illustration of the domain and bond structure for $N=8$.
The four panels show a bulk domain, a length-$N$ domain, a left/right-boundary domain, and a length-zero domain.
Gray boxes mark domains (except for the length-zero domain), yellow squares represent $\ket{t}$, and orange (blue) circles denote $\ket{1}$ ($\ket{0}$).
Red and green solid links indicate type-I and type-II bonds, respectively, while black solid links mark domain boundaries with the surrounding $\ket{t}$ states.
The dashed links indicate the remaining boundary contributions to the domain type factor $c_j$: the orange (blue) dashed link corresponds to a left-boundary domain starting with $\ket{1}$ ($\ket{0}$), whereas the black dashed link corresponds to a right-boundary domain.
(b) Matrix-product representation of $|\Psi_T\rangle$. 
Orange and black bonds denote physical and auxiliary indices, respectively. 
The tensors $K_L$, $K_R$, and $K$ can be absorbed into the local tensors $T^{[j]}$, as indicated by the dashed boxes.
}
\label{fig2}
\end{figure}

The exact NESS reveals an interesting connection to the corresponding closed system. 
For $h_R=h_L=0$, consider 
$Q
\equiv
\lim_{\gamma\rightarrow0}
i(\Omega-\Omega^{\dagger})/\gamma$,
one then finds
\begin{equation}
Q=\sum_{m=0}^{N-1}\frac{1}{h^{m+1}}
\sum_{j=1}^{N-m}
i\hat{\omega}_{2j-1}\hat{\omega}_{2(j+m)},
\label{resolvent symmetry}
\end{equation}
which satisfies $[H_{\mathrm{TFI}},Q]=0$. Here,
$\hat{\omega}_{2j-1}=\prod_{k<j}\sigma_k^{x}\sigma_j^{z}$
and
$\hat{\omega}_{2j}=\prod_{k<j}\sigma_k^{x}\sigma_j^{y}$.
For $h>1$, the coefficients in $Q$ decay exponentially with range,
whereas for $h<1$ they grow exponentially. The critical point $h=1$
is therefore directly reflected in the structure of the NESS.

It is worth noting that the exact construction presented here is not limited
to the specific Lindbladian in Eq.~\eqref{model}.
For certain modifications of the Lindbladian that preserve the HTRS, the same
method can still yield exact NESS solutions.
The extensions and the integrable structure of the NESS are discussed in Appendix~\ref{appB}.

\section{Steady-state Magnetization and Correlations}
\label{Sec4}

\begin{figure*}[t]
  \centering
  \includegraphics[width=\textwidth]{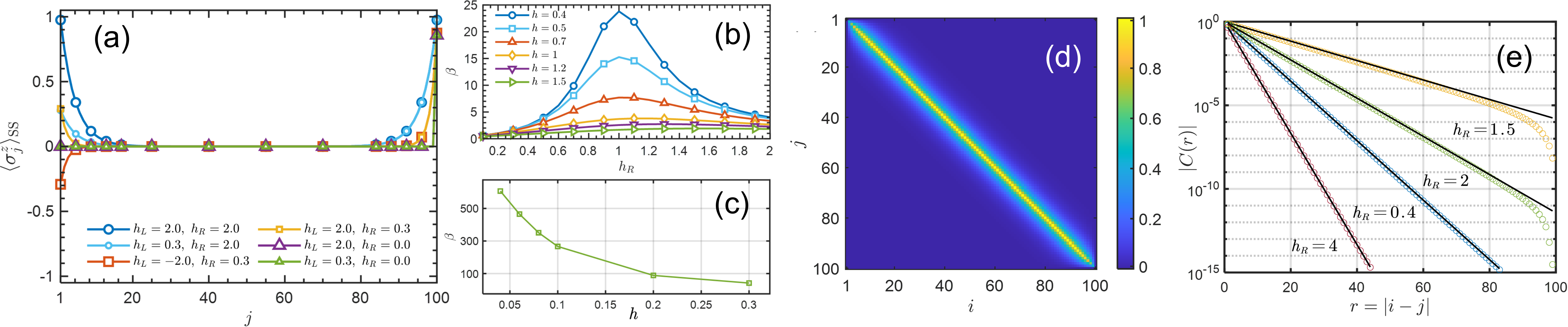}
  \caption{Steady-state magnetization and correlations for generic parameter settings.
(a) Magnetization profiles $\langle\sigma_j^z\rangle_{\rm SS}$ for different values of $h_L$ and $h_R$, indicated by different colors. For $h_R\neq0$, the magnetization is exponentially localized near the boundaries. At $h_R=0$, the magnetization is nonzero only at the rightmost site $N$.
(b) Magnetization localization length $\beta$ as a function of
$h_R$ for several values of the transverse field $h$.
(c) Localization length $\beta$ as a function of $h$ at $h_R=1$,
showing its rapid growth as $h$ decreases.
The common parameters in (b) and (c) are
$h_L=2$, $\gamma=0.1$, and $N=2000$.
(d) Connected correlation function
$C_{ij}=\langle\sigma_i^z\sigma_j^z\rangle_{\rm SS}
-\langle\sigma_i^z\rangle_{\rm SS}
\langle\sigma_j^z\rangle_{\rm SS}$
for $h_L=h_R=2$.
(e) Averaged connected correlation
$|C(r)|$, with
$C(r)=\sum_{|i-j|=r}C_{ij}/(N-r)$,
for $h_L=2$ and different values of $h_R$ indicated by different
colors. The black line shows the exponential decay predicted by the
ratio of the two leading transfer-matrix eigenvalues, $|\xi_2/\xi_1|$.
The common parameters in (a), (d), and (e) are
$N=100$, $\gamma=0.1$, and $h=0.5$.}
  \label{fig3}
\end{figure*}

Eq.~\eqref{MPO_psiT_Tmatrix} is the basis for computation of expectations of observables
\begin{equation}
     \text{Tr}(O\rho_{\rm SS})=\frac{\langle\Psi_T|O|\Psi_T\rangle}{\langle\Psi_T|\Psi_T\rangle} \;.
\label{observable}
\end{equation}
Noting that each local tensor can be decomposed as
$
T^{[j]}=\sum_{S_j} T^{[j]}_{S_j}\ket{S_j},
$
the trace of the unnormalized steady-state density matrix naturally defines a nonequilibrium partition function,
\begin{equation}
    \mathcal{Z}
    =
    \mathbb{T}^{[1]}\mathbb{T}^{[2]}\cdots \mathbb{T}^{[N]},
\label{partition_function}
\end{equation}
where
$
    \mathbb{T}^{[j]}
    =
    \sum_{S_j}
    T^{[j]}_{S_j}\otimes \left(T^{[j]}_{S_j}\right)^*
$
is the local transfer matrix. Any physical observable $O\in \mathrm{End}(\mathcal{H}_A)$ can be expanded in the Pauli-string basis as
$O=\sum_{\{\alpha\}} a_{\{\alpha\}} O_{\{\alpha\}}$ with
$
O_{\{\alpha\}}
=
\sigma^{\alpha_1}\otimes \sigma^{\alpha_2}\otimes\cdots\otimes \sigma^{\alpha_N},
$
where $a_{\{\alpha\}}\in\mathbb{C}$, $\{\alpha\}\equiv\alpha_1,\alpha_2,\ldots,\alpha_N$ and
$\alpha_j\in\{0,x,y,z\}$. 
The steady-state expectation value of $O$ is then given by
\begin{equation}
\langle O\rangle_{\rm SS}
=\sum_{\{\alpha\}}
\frac{
\mathbb{V}^{[1]}_{\alpha_1}
\mathbb{V}^{[2]}_{\alpha_2}
\cdots
\mathbb{V}^{[N]}_{\alpha_N}
}
{\mathcal{Z}},
\label{observables_computation}
\end{equation}
where the operator-inserted matrix is
$
\mathbb{V}^{[j]}_{\alpha_{j}}
=
\sum_{S_j,S'_j}
\langle{S'_j}|\sigma^{\alpha_{j}}\ket{S_j}\,
T^{[j]}_{S_j}\otimes \left(T^{[j]}_{S'_j}\right)^* .
$
% \label{mathbb_V}
% \end{equation}

We now investigate the steady-state $z$-magnetization $\langle\sigma_j^z\rangle_{\rm ss}$
and the connected correlation function
\begin{equation}
C_{ij}
=
\langle\sigma_i^z\sigma_j^z\rangle_{\rm SS}
-
\langle\sigma_i^z\rangle_{\rm SS}
\langle\sigma_j^z\rangle_{\rm SS}.
\label{C_r_def}
\end{equation}
For convenience, in the following we mainly focus on the regime
$h_R\geq 0$.

\subsection{Generic case}
\label{sec:general_case}
For generic parameter choices, the steady-state magnetization exhibits an
exponentially localized profile near the boundaries, as shown in
Fig.~\ref{fig3}(a). To characterize this behavior quantitatively, we
introduce the magnetization localization length $\beta$. 
For a site $j$
close to either site $1$ or site $N$, let $x$ denote its distance from the
nearest boundary. We can parameterize the decay as
\begin{equation}
    \langle \sigma_j^z\rangle_{\rm SS}
    \simeq
    \langle \sigma_{1,N}^z\rangle_{\rm SS}
    e^{-x/\beta},
\end{equation}
where $\beta$ defines the localization length.
With this definition, we will show below that, in the one-sided
dissipative chain considered here, the local field at the
dissipative boundary can exert a decisive influence on the spatial
structure of the entire system.

For fixed $\gamma$ and $h$, Fig.~\ref{fig3}(a) shows that both
$\langle\sigma_N^z\rangle_{\rm SS}$ and the localization length $\beta$
near the right boundary are determined by the strength of the
right-boundary longitudinal field. It is natural to expect local
properties near a boundary to be governed primarily by the terms acting
at that boundary. The behavior near the left boundary, however, is rather unusual. As illustrated by the dark-blue and yellow curves in
Fig.~\ref{fig3}(a), the same value of $h_L$ does not in general produce
the same magnetization profile near the left boundary. This indicates
that the local physics at the left end is strongly influenced by the
longitudinal field applied at the opposite end of the chain.
A direct explanation for this apparently nonlocal dependence can be obtained from the exact solution. 
From Eq.~\eqref{MPO_psiT}, one sees that $h_R$ enters every
bond tensor of the MPS representation of the purified NESS
$|\Psi_T\rangle$, whereas $h_L$ appears only in the left boundary tensor
$K_L$. Consequently, $h_L$ only affects $\langle\sigma_1^z\rangle_{\rm SS}$,
while the localization length $\beta$ at the left boundary is
governed by $h_R$. Furthermore, since the MPS in Eq.~\eqref{MPO_psiT} is
spatially uniform in the bulk, $\beta$ is the same at the two boundaries.

The exact solution can also provide a quantitative expression for $\beta$. Since $|\Psi_T\rangle$ possesses the local-bond structure encoded in Eq.~\eqref{MPO_psiT}, in the large-$N$ limit the local state at a given position can be approximately obtained by
propagating the state at the nearest boundary through successive local
bonds. We introduce the local-state weight vector for a spin pair in the
doubled system at distance $x$ from the nearest boundary,
\begin{equation}
    \bar{P}^{x}
    =
    \left[
    P^{x}_{1},
    P^{x}_{t},
    P^{x}_{0}
    \right]^T ,
\end{equation}
where $P^{x}_{\alpha}$ denotes the weight of the local pair state $|\alpha\rangle$ with $\alpha\in\{1,t,0\}$. According to
Eq.~\eqref{MPO_psiT}, these weights satisfy
\begin{equation}
\bar{P}^{x}
=
\left[
\begin{matrix}
|\lambda_{\rm I}|^2 &
|c_{\rm b}| &
|\lambda_{\rm II}|^2
\\
|c_{\rm b}| &
1 &
|c_{\rm b}|
\\
|\lambda_{\rm II}|^2 &
|c_{\rm b}| &
|\lambda_{\rm I}|^2
\end{matrix}
\right]^x
\bar{P}^{\rm nb},
\end{equation}
where $\bar{P}^{\rm nb}$ denotes the corresponding weight vector at the
nearest boundary. Diagonalizing the transfer matrix above and neglecting
terms that are exponentially suppressed with $N$, together with
$\langle\sigma_j^z\rangle_{\rm SS}\propto P_1^x-P_0^x$ up to an overall
normalization factor, gives
\begin{equation}
    \langle \sigma_j^z\rangle_{\rm SS}
    \propto e^{-x/\beta},
\label{skin_profile_estimate}
\end{equation}
with
\begin{equation}
\beta=
1/
\log\left[
\frac{\Lambda}
{2\bigl(|\lambda_{\rm I}|^2-|\lambda_{\rm II}|^2\bigr)}
\right]
,
\label{localization_length}
\end{equation}
where
\begin{equation}
\Lambda=
|\lambda_{\rm I}|^2+|\lambda_{\rm II}|^2+1
+\sqrt{
|\lambda_{\rm I}|^2+|\lambda_{\rm II}|^2-1
+8c_{\rm b}^2
}.
\end{equation}
Eq.~\eqref{localization_length} shows that the localization length
decreases with increasing $\gamma$ and $h$, whereas its dependence on
$h_R$ is nonmonotonic: $\beta$ first increases and then decreases as
$h_R$ is increased. This behavior is confirmed by the numerical results
in the weak-dissipation regime shown in Fig.~\ref{fig3}(b). In particular,
for $h<1$, the maximum localization length occurs near $h_R=1$.
Furthermore, at $h_R=1$, decreasing the transverse-field strength $h$
leads to a rapid increase of $\beta$, as shown in Fig.~\ref{fig3}(c).
The same behavior follows analytically from
Eq.~\eqref{localization_length}. In the limit $h\rightarrow0$, one finds
\begin{equation}
    \beta
    \simeq
    \left[
    \log\left(
    \frac{h_R^2+1+\gamma^2/4}
    {2|h_R|}
    \right)
    \right]^{-1}.
\label{localization_length_weak_h}
\end{equation}
Eq.~\eqref{localization_length_weak_h} therefore shows that, in the weak-dissipation and weak-transverse-field limit, the localization length diverges at $h_R=\pm1$. The magnetization profile then ceases to be boundary localized, signaling a qualitatively different steady-state structure that we discuss in the next subsection.

Another notable feature of Fig.~\ref{fig3}(a) occurs at $h_R=0$.
In this case, the steady-state magnetization vanishes at every site except the site $N$, irrespective of $\gamma$, $h$, and $h_L$. The origin of this behavior is
again evident from the exact NESS. At $h_R=0$,
\begin{equation}
    |\lambda_{\rm I}|=|\lambda_{\rm II}|,
    \qquad
    |c_+|=|c_-|.
\end{equation}
The first relation removes any preference between locally parallel and
antiparallel spin configurations along the $z$ direction, while the second
assigns equal weights to the two polarizations at the left boundary.
Strikingly, the left-boundary magnetization is therefore completely
insensitive to the left-boundary field. Consequently, all sites except the
directly dissipated one become unbiased between
$|\uparrow\rangle$ and $|\downarrow\rangle$, and the localization
length collapses to $\beta=0$.

The exact solution further enables a direct characterization
of the steady-state correlations. For generic parameters, the connected
two-point correlation function $C_{ij}$ decays exponentially with
the separation $r:=|i-j|$, as shown in Fig.~\ref{fig3}(d). This
behavior follows directly from the transfer-matrix structure of
$|\Psi_T\rangle$. In the bulk, $3\leq j\leq N-2$, the transfer
matrix admits the spectral decomposition
\begin{equation}
\mathbb{T}^{[j]}
=
\sum_n
\xi_n
|n_R\rangle\langle n_L|,
\end{equation}
where the eigenvalues are ordered according to
$|\xi_1|\geq|\xi_2|\geq\cdots$. The two-point function can then be
written schematically as
\begin{equation}
\langle \sigma_i^z\sigma_j^z\rangle_{\rm SS}
\simeq
\frac{
\langle 1_L|
\mathbb{V}^{[i]}_z
\mathbb{T}^{\,r-1}
\mathbb{V}^{[j]}_z
|1_R\rangle
}
{\xi_1^{\,r+1}}.
\label{two_point_transfer}
\end{equation}
At large separations, the leading contributions are
\begin{align}
\langle \sigma_i^z\sigma_j^z\rangle_{\rm SS}
\simeq\;&
\frac{
\langle 1_L|\mathbb{V}^{[i]}_z|1_R\rangle
\langle 1_L|\mathbb{V}^{[j]}_z|1_R\rangle
}
{\xi_1^2}
\nonumber\\
&+
\frac{
\langle 1_L|\mathbb{V}^{[i]}_z|2_R\rangle
\langle 2_L|\mathbb{V}^{[j]}_z|1_R\rangle
}
{\xi_1\xi_2}
\left(
\frac{\xi_2}{\xi_1}
\right)^r
+\cdots .
\label{two_point_asymptotic}
\end{align}
The first term is the factorized contribution
$\langle\sigma_i^z\rangle_{\rm SS}
\langle\sigma_j^z\rangle_{\rm SS}$,
whereas the leading connected contribution is governed by the
subleading eigenvalue $\xi_2$. Consequently, provided that
$|\xi_1|\neq|\xi_2|$, the connected correlation function decays
exponentially as
\begin{equation}
    C_{ij}
    \propto
    \left(
    \frac{\xi_2}{\xi_1}
    \right)^r ,
\end{equation}
up to a parameter-dependent prefactor. 
Fig.~\ref{fig3}(e) shows the averaged connected correlation
$C(r)=\frac{1}{N-r}\sum_{|i-j|=r} C_{ij}$
as a function of the separation $r$. Consistent with the analysis above, $C(r)$ decays exponentially with $r$, with the decay rate governed by $\xi_2/\xi_1$.

The above analysis also applies to $h_R<0$, although the
magnetization profile and correlation differ qualitatively from
those shown in Fig.~\ref{fig3}. For $h_R<0$, one has
$|\lambda_{\rm I}|<|\lambda_{\rm II}|$, so that neighboring spins
preferentially acquire opposite orientations. As a result, both the
magnetization profile and the connected correlations exhibit an
alternating odd--even sign structure. Further details for this regime
are presented in Appendix~\ref{appC}.

\begin{figure}[h]
\centering
\includegraphics[width=\linewidth]{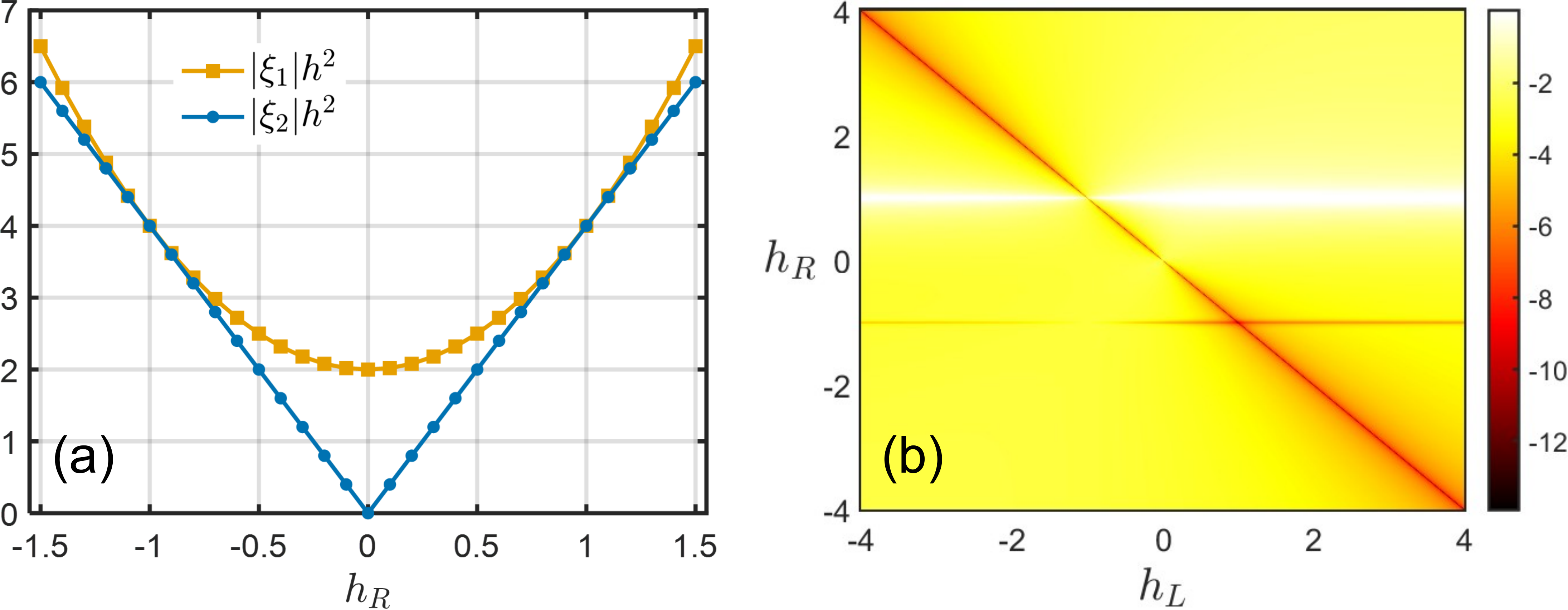}
\caption{Signatures of special steady-state structures in the weak-driving limit of Eq.~\eqref{weak_drive_limit}. Here we set $h=\gamma=10^{-3}$.
(a) Leading eigenvalue magnitudes $|\xi_1|$ and $|\xi_2|$ of $\mathbb{T}$ as functions of $h_R$ for $h_L=2$. The two magnitudes become equal at $h_R=\pm1$.
(b) $\log_{10}\langle m\rangle$ as a function of the boundary longitudinal fields $h_L$ and $h_R$ for $N=500$, where
$\langle m\rangle$ is the steady-state magnetization density.
The magnetization density is strongly enhanced near $h_R=1$ and suppressed near $h_R=-1$. In addition, $h_L=-1$ is distinct from other values of $h_L$ at $h_R=\pm1$.}
\label{fig4}
\end{figure}

\subsection{Long-range correlations in the weak-driving limit}
\label{sec:shock_case}

Eq.~\eqref{two_point_asymptotic} suggests that the above picture
breaks down when $|\xi_1|=|\xi_2|$. In Fig.~\ref{fig4}(a), we show
$|\xi_1|$ and $|\xi_2|$ as functions of $h_R$ in the weak-driving
limit,
\begin{equation}
\gamma\ll1,
\qquad
h\ll1.
\label{weak_drive_limit}
\end{equation}
The two eigenvalues become degenerate in magnitude precisely at
$h_R=\pm1$ in this limit. These are also the points at which the
magnetization localization length given by
Eq.~\eqref{localization_length_weak_h} diverges. Thus, both the
connected correlation and the magnetization profile indicate a
qualitative change in the steady-state structure at $h_R=\pm1$.

A more direct signature of this change is provided by the steady-state
magnetization density. Fig.~\ref{fig4}(b) shows
$
\langle m\rangle
=
\frac{1}{N}\sum_j
\langle\sigma_j^z\rangle_{\rm SS}
$
as a function of $h_L$ and $h_R$ in the weak-driving limit.
Along $h_R=-h_L$, $\langle m\rangle$ is strongly suppressed as a
consequence of the approximate combined reflection--spin-flip
symmetry. More importantly, two pronounced features emerge along
$h_R=\pm1$: the magnetization density is strongly enhanced near
$h_R=1$ and suppressed near $h_R=-1$, clearly distinguishing these
two lines from the generic parameter regime. In addition, along both
$h_R=\pm1$, the behavior at $h_L=-1$ is markedly different from that at other values of $h_L$. This indicates the emergence of a distinct steady-state structure at
$(h_R,h_L)=(\pm1,-1)$. In the following, we focus on the case $h_R=1$. The discussion for $h_R=-1$ is presented in Appendix~\ref{appC}.

We first consider the case $h_R=1$ and $h_L\neq -1$.
From Eq.~\eqref{coeff}, one finds that in the
weak-driving limit,
\[
|\lambda_{\rm I}|
\sim
\mathcal O(1/h),
\qquad
|\lambda_{\rm II}|
\sim
\mathcal O(\gamma/h).
\]
Since $\gamma\ll1$, type-I bonds are overwhelmingly favored in the
steady state. Starting from the dissipatively polarized right boundary,
this preference propagates throughout the chain, such that the dominant
configuration of $|\Psi_T\rangle$ is
$
|11\cdots 11\rangle .
$
Accordingly, all physical spins are polarized in the state
$\ket{\uparrow}$, corresponding to an ideal ferromagnetic NESS
with a spatially uniform magnetization and vanishing connected
correlations.
Fig.~\ref{fig5} shows the numerical results for
$h=\gamma=0.001$. Small deviations from the ideal
ferromagnetic picture are clearly visible. In Fig.~\ref{fig5}(a), the
magnetization profile exhibits an approximately linear spatial
variation rather than remaining perfectly flat, while
Fig.~\ref{fig5}(b) reveals very weak but finite long-range connected
correlations. These deviations originate from subleading configurations
that acquire nonzero weight in the NESS at finite transverse field and
dissipation strength.

\begin{figure}
\centering
\includegraphics[width=\linewidth]{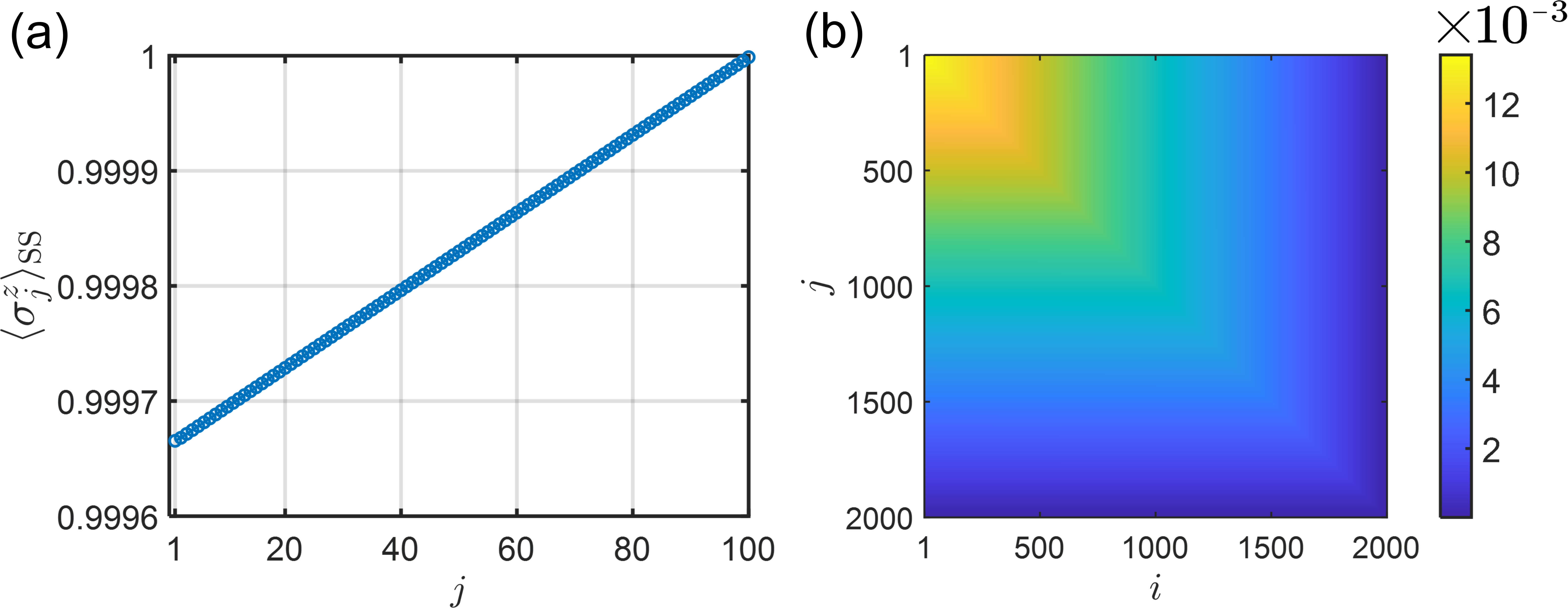}
\caption{Typical steady-state behavior for $h_R=1$ and $h_L\neq-1$ in the weak-driving limit. Here we take $h_L=-2$ and set $h=\gamma=10^{-3}$.
(a) Magnetization profile $\langle\sigma_j^z\rangle_{\rm SS}$ for $N=100$.
(b) Connected correlation function $C_{ij}$ for $N=2000$.
The nearly polarized profile and the nearly vanished connected correlations indicate that the NESS approaches an ideal ferromagnetic state, while the small deviations arise from subleading configurations induced by finite $h$ and $\gamma$.}
\label{fig5}
\end{figure}

\begin{figure*}[t]
\centering
\includegraphics[width=\textwidth]{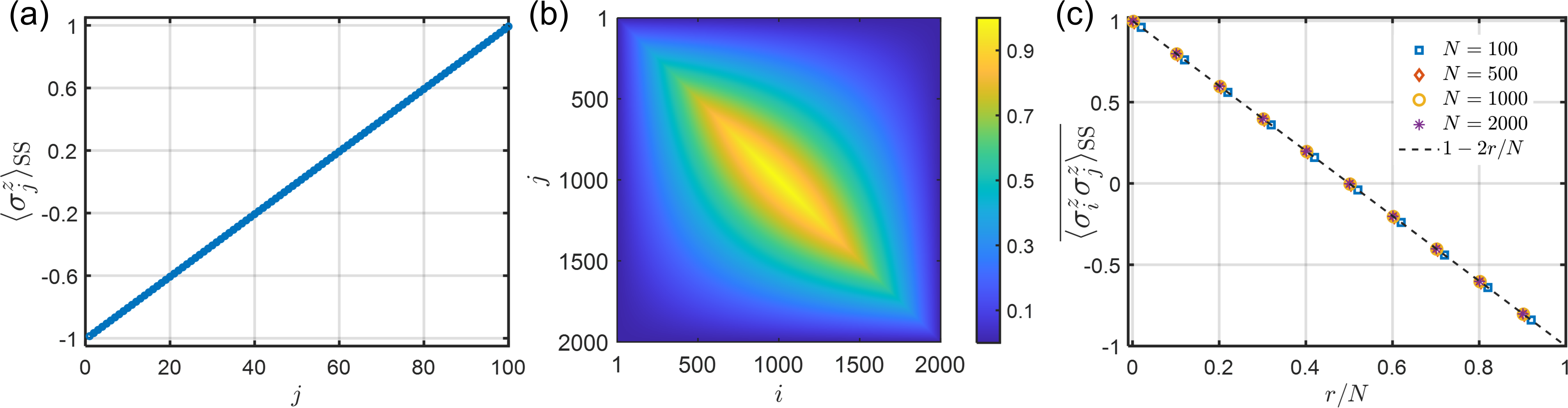}
\caption{Steady-state properties for $h_R=1$ and $h_L=-1$ in the weak-driving limit, where the NESS is dominated by single-interface configurations. Here we set $h=\gamma=10^{-3}$.
(a) Magnetization profile $\langle\sigma_j^z\rangle_{\rm SS}$ as a function of site $j$ for $N=100$, showing a linear spatial dependence.
(b) Connected correlation function
$C_{ij}=\langle\sigma_i^z\sigma_j^z\rangle_{\rm SS}
-\langle\sigma_i^z\rangle_{\rm SS}\langle\sigma_j^z\rangle_{\rm SS}$
as a function of sites $i$ and $j$ for $N=2000$, exhibiting long-range correlations.
(c) Distance-averaged two-point correlation 
$\overline{\langle\sigma_i^z\sigma_j^z\rangle}_{\rm SS}$ as a function of $r/N$
for different system sizes $N$, where the overline denotes the average over all pairs $(i,j)$ with fixed $|i-j|=r$.
The dashed lines denote the theoretical prediction from the exact solution.}
\label{fig6}
\end{figure*}

The subleading configurations discussed above give only small corrections to the dominant
steady-state structure for $h_L\ne-1$.
At $h_L=-1$, however, they become dominant and qualitatively alter the
steady-state properties. This behavior is illustrated in Fig.~\ref{fig6}
for $(h_R,h_L)=(1,-1)$. In contrast to the nearly ferromagnetic state in
Fig.~\ref{fig5}, the magnetization in Fig.~\ref{fig6}(a) varies
approximately linearly from
$\langle\sigma_1^z\rangle_{\rm SS}\simeq-1$ to
$\langle\sigma_N^z\rangle_{\rm SS}\simeq1$, accompanied by pronounced
long-range correlations [Fig.~\ref{fig6}(b)].

The microscopic origin of this behavior follows directly from
Eq.~\eqref{coeff}. In the weak-driving limit, $h_L=-1$ gives
\[
|c_+|
\sim
\mathcal O(\gamma/h),
\qquad
|c_-|
\sim
\mathcal O(1/h),
\]
so that the left boundary strongly favors the local state $\ket{0}$.
By contrast, the dissipatively polarized right boundary, together with
the strong preference for type-I bonds, favors uniform polarization in
$\ket{1}$. This competition can be accommodated at minimal cost by
introducing either a single type-II bond connecting $\ket{0}$ and
$\ket{1}$ or a single $\ket{t}$ state. We refer to this minimal structure
separating the uniformly $\ket{0}$ and $\ket{1}$ regions as an interface.
Since each additional interface strongly suppresses the weight of the
corresponding configuration in $|\Psi_T\rangle$, single-interface
configurations dominate in the weak-driving limit. Moreover,
Eq.~\eqref{MPO_psiT} shows that $|\Psi_T\rangle$ is spatially uniform in
the bulk, leaving no preferred interface position. The interface is
therefore uniformly distributed along the chain, and
\begin{equation}
\begin{aligned}
|\Psi_T\rangle
\simeq
\frac{1}{\sqrt{2N}}
&\sum_{j=1}^{N}
\Big(
|0_10_2\cdots0_{j-1}1_j1_{j+1}\cdots1_N\rangle\\
&\qquad+
|0_10_2\cdots0_{j-1}t_j1_{j+1}\cdots1_N\rangle
\Big).
\end{aligned}
\label{single-kink-state}
\end{equation}
The single-interface mechanism found here has close analogues in other
one-dimensional statistical systems. In a classical Ising chain with
competing boundary fields, the zero-temperature transition is
characterized by single-kink configurations~\cite{Chancellor2022Ising}, while on
the coexistence line of the open ASEP the steady state contains a
delocalized shock separating two regions with different densities
~\cite{Derridat1993,Schutz1998,Schutz2004}.

Eq.~\eqref{single-kink-state} immediately yields
\begin{equation}
\langle\sigma_j^z\rangle_{\rm SS}
\simeq
\frac{2j}{N}-1,
\label{linear_profile}
\end{equation}
reproducing the linear profile in Fig.~\ref{fig6}(a).
Eq.~\eqref{single-kink-state} also determines the two-point correlations. For a given
single-interface configuration, $\sigma_i^z\sigma_j^z=-1$ if the
interface lies between sites $i$ and $j$, and
$\sigma_i^z\sigma_j^z=+1$ otherwise. Averaging uniformly over the
interface position gives
\begin{equation}
\langle\sigma_i^z\sigma_j^z\rangle_{\rm SS}
\simeq
1-2\frac{r}{N},
\label{2-point-correlator}
\end{equation}
where $r=|i-j|$. Fig.~\ref{fig6}(c) shows the distance-averaged
correlation
$
\overline{\langle\sigma_i^z\sigma_j^z\rangle}_{\rm SS}
:=
\sum_{|i-j|=r}\langle\sigma_i^z\sigma_j^z\rangle_{\rm SS}
/
\sum_{|i-j|=r}1
$
as a function of $r/N$ for several $N$. The data collapse onto
Eq.~\eqref{2-point-correlator}, in excellent agreement with the
single-interface picture.

The single-interface description becomes asymptotically exact in the
weak-driving limit and remains accurate for sufficiently small systems
at finite $h$ and $\gamma$. For any finite $h$ and $\gamma$, subleading configurations acquire nonzero weights and become increasingly important as $N$ increases. Consequently, the simple dominant-configuration picture eventually receives appreciable corrections for sufficiently large systems. A more detailed analysis of these corrections and the system-size regime over which the single-interface picture remains valid at finite $\gamma$ and $h$ is presented in Appendix~\ref{appD}.

\section{Summary}\label{Sec5}
In this work, we obtain an exact NESS of a boundary-driven TFI chain with
dissipation acting at a single boundary. Exploiting HTRS, we map the
steady-state problem onto a pure NESS of an associated doubled system and
show that this state admits a simple local-bond structure. This structure
naturally leads to a nonequilibrium partition function and reveals how a
local field at the dissipative boundary can control the spatial organization
of the NESS throughout the chain.
For generic parameters, the NESS exhibits boundary-localized
magnetization and short-ranged correlations, with the associated length
scales governed by the field at the dissipative boundary. In the
weak-driving limit, suitably tuned boundary fields reorganize the NESS into a delocalized
single-interface structure, producing a linear magnetization
profile and long-range correlations. 
Our results provide an exact example of obtaining the NESS via HTRS (equivalently, KMS detailed balance), and demonstrate how boundary control can govern global properties in an open quantum system.

\begin{acknowledgments}
We thank Xueliang Wang, Yu-Guo Liu and Mingchen Zheng for helpful discussions.
This work is supported by National Key Research and Development Program of China
(Grant No. 2023YFA1406704) and the NSFC under Grants No. 12474287, No. 12547107, and No. T2121001.
\end{acknowledgments}

\appendix

\section{Proof of steady-state uniqueness}
\label{appA}
In this Appendix, we prove that the Lindblad dynamics considered in the
main text possesses a unique steady state for $N\ge 2$, $h\neq0$ and $\gamma>0$.

According to the criterion of Yoshida~\cite{Yoshida2024}, it is sufficient
to show that
\begin{equation}
    G_N=H_N-iL^\dagger L,
    \qquad L=\sqrt{\gamma}\sigma_N^+,
\end{equation}
generate the full operator algebra $\mathcal{B}(\mathcal H_N)$. By Burnside's
theorem~\cite{Burnside,Burnside2}, this is equivalent to showing that
$G_N$ and $L$ have no common nontrivial invariant subspace. Since
$L^\dagger L=\gamma(1-\sigma_N^z)/2$, a scalar shift allows us to replace
$G_N$ by
\begin{equation}
    \widetilde G_N=H_N+\frac{i\gamma}{2}\sigma_N^z .
\end{equation}
Writing $\mathcal H_N=\mathcal H_{N-1}\otimes\mathbb C^2$ in the
$\sigma_N^z$ basis gives
\begin{equation}
    \widetilde G_N=
    \begin{bmatrix}
        K_{N-1}+q & hI\\
        hI & K_{N-1}-q
    \end{bmatrix},
    \qquad
    L=
    \begin{bmatrix}
        0&I\\
        0&0
    \end{bmatrix},
\label{eq:uniqueness_block}
\end{equation}
where
\begin{align}
    K_{N-1}
    &=
    \sum_{j=1}^{N-2}\sigma_j^z\sigma_{j+1}^z
    +h\sum_{j=1}^{N-1}\sigma_j^x
    +h_L\sigma_1^z,\\
    q&=\sigma_{N-1}^z+\beta I,
    \qquad
    \beta=h_R+\frac{i\gamma}{2}.
\end{align}
Since multiplication by a nonzero scalar does not affect invariant
subspaces, we suppress the overall factor $\sqrt{\gamma}$ in $L$ below.

Suppose that $S\subseteq\mathcal H_N$ is invariant under both
$\widetilde G_N$ and $L$. Define
\begin{align}
    U_0&=\{u:(u,0)^T\in S\},\\
    U_1&=\{d:\exists\,u\ {\rm such\ that}\ (u,d)^T\in S\}.
\end{align}
The invariance under $L$ gives $U_1\subseteq U_0$, while
$\widetilde{G}_N$ applied to $(u,0)^T$, with $u\in U_0$,
gives $U_0\subseteq U_1$ for $h\neq0$. Hence
$U_0=U_1\equiv U$.
For $u\in U$, the lower component of
$[\widetilde G_N]^{2}(u,0)^T$ is $2hK_{N-1}u$. Therefore,
\begin{equation}
    K_{N-1}U\subseteq U.
\label{eq:K_U}
\end{equation}

Moreover,
\begin{equation}
    \begin{bmatrix}qu\\hu\end{bmatrix}
    =
    \widetilde G_N
    \begin{bmatrix}u\\0\end{bmatrix}
    -
    \begin{bmatrix}K_{N-1}u\\0\end{bmatrix}
    \in S.
\label{eq:qhu}
\end{equation}
Applying $\widetilde G_N$ to Eq.~\eqref{eq:qhu} and subtracting
Eq.~\eqref{eq:qhu} with $u$ replaced by $K_{N-1}u$, we obtain
\begin{equation}
    \bigl([K_{N-1},q]+2\beta q\bigr)U\subseteq U,
\label{eq:q_condition}
\end{equation}
where $q^2=2\beta q+(1-\beta^2)I$ has been used.
Let $P$ denote the orthogonal projector onto $U$ and $Q=I-P$.
Since $K_{N-1}$ is Hermitian, Eq.~\eqref{eq:K_U} implies
$[P,K_{N-1}]=0$. Defining $X=QqP$, Eq.~\eqref{eq:q_condition}
then yields
\begin{equation}
    K_QX-XK_U+2\beta X=0,
\label{eq:sylvester}
\end{equation}
where $K_U$ and $K_Q$ are the restrictions of $K_{N-1}$ to
$U$ and $U^\perp$, respectively. Multiplying
Eq.~\eqref{eq:sylvester} by $X^\dagger$ and taking the imaginary
part of the trace gives
\begin{equation}
    \gamma\;\text{Tr}(X^\dagger X)=0,
\end{equation}
since $K_U$ and $K_Q$ are Hermitian and
$\operatorname{Im}\beta=\gamma/2$. Thus $X=0$, so that
$qU\subseteq U$, and consequently
\begin{equation}
    \sigma_{N-1}^zU\subseteq U.
\label{eq:z_U}
\end{equation}

It remains to use the following elementary result.

\begin{lemma}
For $h\neq0$, the operators
\begin{equation}
    K_n=
    \sum_{j=1}^{n-1}\sigma_j^z\sigma_{j+1}^z
    +h\sum_{j=1}^{n}\sigma_j^x
    +h_L\sigma_1^z
\end{equation}
and $\sigma_n^z$ have no common nontrivial invariant subspace.
\label{lemma:uniqueness}
\end{lemma}

\begin{proof}
The statement follows by induction. For $n=1$,
$K_1=h\sigma_1^x+h_L\sigma_1^z$, which shares no nontrivial
invariant subspace with $\sigma_1^z$ when $h\neq0$.

Assume the statement holds for $n-1$. If
$W\subseteq\mathcal H_n$ is invariant under both $K_n$ and
$\sigma_n^z$, then
\begin{equation}
    W=(W_+\otimes|\uparrow\rangle)
    \oplus(W_-\otimes|\downarrow\rangle).
\end{equation}
Using
\begin{equation}
    K_n=
    \begin{bmatrix}
        K_{n-1}+\sigma_{n-1}^z & hI\\
        hI & K_{n-1}-\sigma_{n-1}^z
    \end{bmatrix},
\end{equation}
the off-diagonal blocks imply $W_+=W_-\equiv W_0$ for $h\neq0$,
while the diagonal blocks imply
\begin{equation}
    K_{n-1}W_0\subseteq W_0,
    \qquad
    \sigma_{n-1}^zW_0\subseteq W_0.
\end{equation}
The induction hypothesis therefore gives
$W_0=\{0\}$ or $W_0=\mathcal H_{n-1}$, completing the proof.
\end{proof}

Eqs.~\eqref{eq:K_U} and \eqref{eq:z_U}, together with
Lemma~\ref{lemma:uniqueness}, imply that
$U=\{0\}$ or $U=\mathcal H_{N-1}$.
In the former case $S=\{0\}$. In the latter,
$(u,0)^T\in S$ for every $u\in\mathcal H_{N-1}$, and
Eq.~\eqref{eq:uniqueness_block} then implies
$(0,u)^T\in S$ for every $u$, so that $S=\mathcal H_N$.
Thus $\widetilde G_N$ and $L$ have no common nontrivial invariant
subspace. Burnside's theorem therefore gives the full operator algebra,
and the criterion of Ref.~\cite{Yoshida2024} guarantees a unique and full-rank NESS.

\section{Derivation and Generalization of the Exact Solution}
\label{appB}
\subsection{Proof of Eq.~\eqref{coeff}}\label{appB1}
In this section, we derive the central result of the main text,
Eq.~\eqref{coeff}, from Eqs.~\eqref{specific_constraints_0}--\eqref{parameter_boundary}.

Eq.~\eqref{specific_constraints_0} immediately yields the following
three relations:
\begin{align}
    &\operatorname{Coeff}(\dots t_{j-1}0_jt_{j+1}\dots)
    =
    \operatorname{Coeff}(\dots t_{j-1}1_jt_{j+1}\dots),
    \label{lemma0_1}
    \\
    &\operatorname{Coeff}(\dots 0_{j-1}0_j1_{j+1}\dots)
    =
    \operatorname{Coeff}(\dots 0_{j-1}1_j1_{j+1}\dots),
    \label{lemma0_2}
    \\
    &\operatorname{Coeff}(\dots 1_{j-1}0_j0_{j+1}\dots)
    =
    \operatorname{Coeff}(\dots 1_{j-1}1_j0_{j+1}\dots).
    \label{lemma0_3}
\end{align}
Here, ``$\dots$'' indicates that all sites other than $j-1$, $j$, and
$j+1$ are identical on the two sides of each equality, whereas ``$\cdots$''
is used as the usual ellipsis. Combining Eqs.~\eqref{lemma0_2} and
\eqref{lemma0_3}, we further obtain
\begin{equation}
    \operatorname{Coeff}(\dots \underset{j}{0}11\cdots10\dots)
    =
    \operatorname{Coeff}(\dots \underset{j}{0}100\cdots0\dots).
    \label{lemma0_4}
\end{equation}
Site labels are omitted below whenever no ambiguity arises.
For convenience, we can introduce the following definition.

\textit{Definition} 3 (cluster).
For a basis state $\ket{S_1S_2\cdots S_N}$, a consecutive sequence of
$l$ sites is called a cluster if
\begin{equation}
    S_j=1,
    \qquad
    j=k,k+1,\ldots,k+l-1,
\end{equation}
and it is bounded by $S_{k-1}=S_{k+l}=0$.

We first consider the coefficient of a basis state
$\ket{S_1S_2\cdots S_N}$ containing a single domain in the expansion of
$\ket{\Psi_T}$. For convenience, we introduce the notation
\begin{equation}
    \left(j\mid S_jS_{j+1}\cdots S_{j+k}\right)
    :=
    \operatorname{Coeff}
    \left(
    t\cdots t
    S_jS_{j+1}\cdots S_{j+k}
    t\cdots t
    \right),
\end{equation}
where $S_j,S_{j+1},\ldots,S_{j+k}\neq t$.

Consider an arbitrary single-domain basis state that is neither uniformly
$\ket{0}$ nor uniformly $\ket{1}$ and contains $n\geq 0$ clusters of
arbitrary lengths. Its coefficient in $\ket{\Psi_T}$ is equal
to
\begin{equation}
    (
    j\mid
    \alpha
    \underbrace{0101\cdots01}_{\substack{n\ \text{copies}\\
    \text{of ``01''}}}
    00\cdots0\beta
       ),
    \label{lemma_domain}
\end{equation}
where $\alpha,\beta\in\{0,1\}$ specify the states at the two ends of the
domain, and $j$ denotes its starting site. Eq.~\eqref{lemma_domain}
follows directly from repeated application of
Eqs.~\eqref{lemma0_2}--\eqref{lemma0_4}. 

We next establish the following lemma.
\begin{lemma}
Consider two basis states appearing in the expansion of $\ket{\Psi_T}$ that
are identical except for their rightmost domains, both of length $n$. If the
two domains are bulk domains and contain the same numbers of type-I and type-II bonds, then the two basis states have equal coefficients in the expansion of $\ket{\Psi_T}$.
\label{lemma2}
\end{lemma}
\begin{proof}
We first prove the lemma for single-domain states. The case $n=1$
follows immediately from Eq.~\eqref{lemma0_1}. We next consider $n=2$.
Eq.~\eqref{specific_constraints_0} gives
\begin{equation}
\begin{aligned}
    h(j|11)-h(j|01)-\sqrt{2}(j+1|1)&=0,\\
    h(j|00)-h(j|10)-\sqrt{2}(j+1|0)&=0,\\
    h(j|11)-h(j|10)-\sqrt{2}(j|1)&=0,\\
    h(j|00)-h(j|01)-\sqrt{2}(j|0)&=0.
\end{aligned}
\label{lemma2_n=2_s1}
\end{equation}
Solving these four equations yields
\begin{equation}
    (j|11)=(j|00),\quad
    (j|10)=(j|01),\quad
    (j|0)=(j+1|1).
\label{lemma2_n=2_1}
\end{equation}
Eq.~\eqref{specific_constraints_0} also gives
\begin{equation}
    h(j+1|11)-h(j+1|01)-\sqrt{2}(j+2|1)=0.
\end{equation}
Together with Eq.~\eqref{lemma2_n=2_s1}, this implies
$(j+1|11)=(j|11)+c$ and $(j+1|01)=(j|01)+c$, for some $c\in\mathbb{C}$. The remaining translations of length-$2$ domains satisfy analogous relations. We now show that $c=0$. Eq.~\eqref{specific_constraints_0} gives
\begin{align}
    h(j|111)-h(j|110)-\sqrt{2}(j|11)&=0,
    \label{lemma2_n=2_s2}\\
    h(j|111)-h(j|011)-\sqrt{2}(j+1|11)&=0,
    \label{lemma2_n=2_s3}\\
    h(j|000)-h(j|001)-\sqrt{2}(j|00)&=0,
    \label{lemma2_n=2_s4}\\
    h(j|000)-h(j|100)-\sqrt{2}(j+1|00)&=0.
    \label{lemma2_n=2_s5}
\end{align}
Subtracting Eq.~\eqref{lemma2_n=2_s3} from
Eq.~\eqref{lemma2_n=2_s2} and using
$(j+1|11)=(j|11)+c$, we obtain
$(j|011)=(j|110)+c$.
Similarly, subtracting Eq.~\eqref{lemma2_n=2_s5} from
Eq.~\eqref{lemma2_n=2_s4} gives
$(j|001)=(j|100)-c$.
Eqs~\eqref{lemma0_2} and \eqref{lemma0_3} identify the coefficients on
the left- and right-hand sides of these two relations, respectively.
Therefore, $c=0$. Increasing $k$ and performing the above process repeatedly, it is easy to prove that
\begin{equation}
\begin{aligned}
    &(j|11)=(j|00)=(k|11)=(k|00),\\
    &(j|10)=(j|01)=(k|10)=(k|01).
\end{aligned}
\label{lemma2_n=2_2}
\end{equation}
for all $j,k=2,3,\ldots,N-2$.

Next, we assume that the lemma holds for $n<m$ and prove that it also holds
for $n=m$. We first consider domains at a fixed position $j$.
Suppose that the $m$-domain contains $r$ type-II bonds, with $r<m$.
Since Eq.~\eqref{lemma_domain} has already been established, for $r>0$ it
suffices to prove that the coefficients corresponding to the endpoint pairs
$(1,1)$ and $(0,0)$ are equal, and likewise that those corresponding to
$(0,1)$ and $(1,0)$ are equal.

According to Eq.~\eqref{specific_constraints_0}, we have
\begin{widetext}
\begin{align}
    &(j|1\;\underbrace{01010\cdots}_{\substack{\frac{r-1}{2}\ \text{copies} \\ \text{of ``01''}}}
    \underbrace{00\cdots0}_{\substack{m-r\ \\\text{copies of ``0''}}})
    -(j|0\;\underbrace{01010\cdots}_{\substack{\frac{r-1}{2}\ \text{copies} \\ \text{of ``01''}}}
    \underbrace{00\cdots0}_{\substack{m-r\ \\\text{copies of ``0''}}})
    +\frac{\sqrt{2}}{h}(j{+}1|
    \underbrace{01010\cdots}_{\substack{\frac{r-1}{2}\ \text{copies} \\ \text{of ``01''}}}
    \underbrace{00\cdots0}_{\substack{m-r\ \\\text{copies of ``0''}}})=0,
    \label{lemma2_s1}\\
    &(j|\underbrace{00\cdots0}_{\substack{m-r\ \\\text{copies of ``0''}}}
    \underbrace{1010\cdots}_{\substack{\frac{r-1}{2}\ \text{copies} \\ \text{of ``10''}}}\;1)
    -(j|\underbrace{00\cdots0}_{\substack{m-r\ \\\text{copies of ``0''}}}
    \underbrace{1010\cdots}_{\substack{\frac{r-1}{2}\ \text{copies} \\ \text{of ``10''}}}\;0)
    +\frac{\sqrt{2}}{h}(j|
    \underbrace{00\cdots}_{\substack{m-r\ \text{copies} \\ \text{of ``0''}}}
    \underbrace{1010\cdots}_{\substack{\frac{r-1}{2}\ \\\text{copies of ``10''}}})=0.
    \label{lemma2_s2}
\end{align}
\end{widetext}
By Eq.~\eqref{lemma_domain}, the second terms in these two expressions are
equal. Since the lemma holds for domains of length $m-1$, the third terms
are also equal. The first terms must therefore be equal, which proves the
claim for the endpoint pairs $(0,1)$ and $(1,0)$. The proof for the endpoint
pairs $(0,0)$ and $(1,1)$ when $r>0$, as well as the case $r=0$, proceeds
analogously. By mathematical induction, Lemma~\ref{lemma2} holds for single-domain states with domains at the same position.

We next show that Lemma~\ref{lemma2} also holds for single-domain states at
different positions. Following the argument leading to
Eq.~\eqref{lemma2_n=2_2}, one can show that, for any
$j,k=2,3,\ldots,N-m$,
\begin{equation}
\begin{aligned}
    &(j \mid \underbrace{000\cdots 0}_{m\ \text{copies of ``0''}})
    =
    (k \mid \underbrace{000\cdots 0}_{m\ \text{copies of ``0''}}),\\
    &(j \mid \underbrace{111\cdots 1}_{m\ \text{copies of ``1''}})
    =
    (k \mid \underbrace{111\cdots 1}_{m\ \text{copies of ``1''}}).
\end{aligned}
\label{lemma2_s3}
\end{equation}
Using Eq.~\eqref{specific_constraints_0}, we obtain
\begin{widetext}
\begin{align}
    &(j\mid \underbrace{00\cdots0}_{\substack{m\text{ copies}\\ \text{of ``0''}}})
    -(j\mid \underbrace{00\cdots0}_{\substack{m-1\text{ copies}\\ \text{of ``0''}}}1)
    -\frac{\sqrt{2}}{h}(j\mid
    \underbrace{00\cdots0}_{\substack{m-1\text{ copies}\\ \text{of ``0''}}})=0,
    \label{lemma2_s4}\\
    &(k\mid \underbrace{00\cdots0}_{\substack{m\text{ copies}\\ \text{of ``0''}}})
    -(k\mid \underbrace{00\cdots0}_{\substack{m-1\text{ copies}\\ \text{of ``0''}}}1)
    -\frac{\sqrt{2}}{h}(k\mid
    \underbrace{00\cdots0}_{\substack{m-1\text{ copies}\\ \text{of ``0''}}})=0.
    \label{lemma2_s5}
\end{align}
\end{widetext}
The first terms in Eqs.~\eqref{lemma2_s4} and \eqref{lemma2_s5} are equal by
Eq.~\eqref{lemma2_s3}, while the third terms are equal by the induction
hypothesis. Hence, the second terms are also equal. Substituting this result
back into Eq.~\eqref{specific_constraints_0} and repeatedly applying the same
argument, one can obtain
\begin{equation}
    (j\mid \underbrace{00\cdots1}_{\substack{r\text{ type-II bonds}}})
    =
    (k\mid \underbrace{00\cdots1}_{\substack{r\text{ type-II bonds}}}).
\label{lemma2_s6}
\end{equation}
The cases with other endpoints pairs follow in the same way. This completes
the proof of Lemma~\ref{lemma2} for single-domain states.

Since the above argument does not involve any sites to the left of the domain under consideration, it therefore extends directly to the case in which additional,
identical domains appear to the left of the domain.

\end{proof}

Lemma~\ref{lemma2} strongly suggests that the coefficient of a basis state
in the expansion of $\ket{\Psi_T}$ depends only on the structures of its
domains. We therefore introduce the following notation.

Consider an arbitrary basis state containing $k$ domains. From left to
right, let their starting positions be $x_1,x_2,\ldots,x_k$. The numbers of
type-I and type-II bonds in the $i$th domain are denoted by $a_i$ and $b_i$,
respectively, and its length by $l_i$. We first introduce the notation for
the case in which all domains are bulk domains, and subsequently extend it
to boundary domains. The coefficient of this basis state is denoted by
\begin{equation}
\bigl(
x_1,x_2,\ldots,x_k
\mid
\{a_1,b_1\},\ldots,\{a_k,b_k\}
\bigr).
\label{domain_coeff_notation}
\end{equation}
In particular, if $x_1,x_2,\ldots,x_k\in\varnothing$, then
\[
(x_1,x_2,\ldots,x_k\mid\{a_1,b_1\},\ldots,\{a_k,b_k\})=1.
\]

Using this notation, we establish the following lemma.

\begin{lemma}
Consider an arbitrary basis state appearing in the expansion of
$\ket{\Psi_T}$, with coefficient
\[
\bigl(
x_1,x_2,\ldots,x_k
\mid
\{a_1,b_1\},\ldots,\{a_k,b_k\}
\bigr).
\]
Using $\lambda_{\rm I}$, $\lambda_{\rm II}$, $c_{b}$ and $c_{\pm}$ defined by Eq.~\eqref{bond_weight} and \eqref{domain_factors}, the following factorization relations hold.

\text{(1)} If the rightmost domain is a right-boundary domain, then
\begin{equation}
\begin{aligned}
&
\bigl(
x_1,\ldots,x_k
\mid
\{a_1,b_1\},\ldots,\{a_k,b_k\}
\bigr)
\\
& =
2c_{\rm b}\lambda_{\rm I}^{a_k}\lambda_{\rm II}^{b_k}
\bigl(
x_1,\ldots,x_{k-1}
\mid
\{a_1,b_1\},\ldots,\{a_{k-1},b_{k-1}\}
\bigr).
\end{aligned}
\label{lemma3_1}
\end{equation}

\text{(2)} If the rightmost domain is a bulk domain, then
\begin{equation}
\begin{aligned}
&
\bigl(
x_1,\ldots,x_k
\mid
\{a_1,b_1\},\ldots,\{a_k,b_k\}
\bigr)
\\
& =
c_{\rm b}\lambda_{\rm I}^{a_k}\lambda_{\rm II}^{b_k}
\bigl(
x_1,\ldots,x_{k-1}
\mid
\{a_1,b_1\},\ldots,\{a_{k-1},b_{k-1}\}
\bigr).
\end{aligned}
\label{lemma3_2}
\end{equation}

\text{(3)} If the leftmost domain is a left-boundary domain, then
\begin{equation}
\begin{aligned}
&
\bigl(
x_1,\ldots,x_k
\mid
\{a_1,b_1\},\ldots,\{a_k,b_k\}
\bigr)
\\
& =
c_{\pm}\lambda_{\rm I}^{a_1}\lambda_{\rm II}^{b_1}
\bigl(
x_2,\ldots,x_k
\mid
\{a_2,b_2\},\ldots,\{a_k,b_k\}
\bigr).
\end{aligned}
\label{lemma3_3}
\end{equation}

\rm{(4)} If the basis state only has a length-$N$ domain, then 
\begin{equation}
\bigl(
1
\mid
\{a_1,b_1\}
\bigr)
 =
2c_{\pm}\lambda_{\rm I}^{a_1}\lambda_{\rm II}^{b_1}.
\label{lemma3_4}
\end{equation}
Here, $c_{\pm}=c_+$ for a left-boundary domain beginning with $\ket{1}$,
whereas $c_{\pm}=c_-$ for one beginning with $\ket{0}$.

\label{lemma3}
\end{lemma}
\begin{proof}
We first prove statements (1) and (2). We begin with a single-domain state,
corresponding to $k=1$, and denote the domain length by $l$.
When $l=1$, Eqs.~\eqref{specific_constraints_0} and
\eqref{specific_constraints_boundary} give
\begin{align}
    &(N-1|01)=2\lambda_{\rm II} (N-1|0),\label{lemma3_s1}\\
    &(N-1|11)=2\lambda_{\rm I} (N-1|1),\label{lemma3_s2}\\
    &h(N-1|01)-h(N-1|11)+\sqrt{2} (N|1)=0,\label{lemma3_s3}\\
    &(N|1)=2c_{\rm b}\;\text{Coeff}(tt\cdots t).\label{lemma3_s4}
\end{align}
Noting that $(N-1|0)=(N-1|1)$ by Eq.~\eqref{lemma0_1}, and substituting
Eqs.~\eqref{lemma3_s1} and \eqref{lemma3_s2} into
Eq.~\eqref{lemma3_s3}, we obtain
\begin{equation}
\begin{aligned}
    &(N-1|01)=\lambda_{\rm II}(N|1),\\
    &(N-1|11)=\lambda_{\rm I}(N|1),\\
    &(N-1|1)=\frac{1}{2}(N|1).
\end{aligned}
\end{equation}
Finally, using Eq.~\eqref{lemma3_s4} together with
$\text{Coeff}(tt\cdots t)=1$, and noting from Lemma~\ref{lemma2} that
translating a bulk domain does not change the coefficient of the
corresponding basis state, we complete the proof for $k=1$ and $l=1$.

The case $k=1$ and $l>1$ can be proved by mathematical induction. Assume
that Lemma~\ref{lemma3} holds for all single-domain states with $l<m$. We
now prove it for $l=m$. For notational simplicity, we take $m$ to be odd.
When the domain contains the maximal number of type-II bonds,
$b_k=m-1$, Eqs.~\eqref{specific_constraints_0} and
\eqref{specific_constraints_boundary} give us
\begin{widetext}
\begin{align}
    &(N-m|\underbrace{0101\cdots01}_{\substack{\frac{m+1}{2} \text{ copies} \\ \text{of ``01''}}})=2\lambda_{\rm II} (N-m|\underbrace{0101\cdots01}_{\substack{\frac{m-1}{2} \text{ copies} \\\text{ of ``01''}}}0),\label{lemma3_s5}\\
    &(N-m|\underbrace{1010\cdots10}_{\substack{\frac{m-1}{2} \text{ copies}\\\text{of ``10''}}}11)=2\lambda_{\rm I} (N-m|\underbrace{1010\cdots10}_{\substack{\frac{m-1}{2} \text{ copies}\\\text{of ``10''}}}1),\label{lemma3_s6}\\
    &(N-m|\underbrace{0101\cdots01}_{\substack{\frac{m+1}{2} \text{ copies} \\\text{of ``01''}}})
    -(N-m|11\underbrace{0101\cdots01}_{\substack{\frac{m-1}{2} \text{ copies} \\\text{of ``01''}}})
    +\frac{\sqrt{2}}{h} (N-m+1|1\underbrace{0101\cdots01}_{\substack{\frac{m-1}{2} \text{ copies} \\\text{of ``01''}}})=0.\label{lemma3_s7}
\end{align}
\end{widetext}
And Lemma~\ref{lemma2} gives
\begin{align}
    &(N-m|\underbrace{1010\cdots10}_{\substack{\frac{m-1}{2}\ \text{ copies} \\\text{of ``10''}}}11)=(N-m|11\underbrace{0101\cdots01}_{\substack{\frac{m-1}{2}\ \text{ copies} \\\text{of ``01''}}}),\label{lemma3_s8}\\
    &(N-m|\underbrace{0101\cdots01}_{\substack{\frac{m-1}{2}\ \text{ copies} \\\text{of ``01''}}}0)=(N-m|\underbrace{1010\cdots10}_{\substack{\frac{m-1}{2}\ \text{ copies} \\\text{of ``10''}}}1).\label{lemma3_s9}
\end{align}
Substituting Eqs.~\eqref{lemma3_s5}, \eqref{lemma3_s6},
\eqref{lemma3_s8}, and \eqref{lemma3_s9} into
Eq.~\eqref{lemma3_s7}, we obtain the key intermediate result
\begin{equation}
    (N-m|\underbrace{1010\cdots10}_{\substack{\frac{m-1}{2}\ \text{copies} \\\text{of ``10''}}}1)=\frac{1}{2}(N-m+1|\underbrace{1010\cdots10}_{\substack{\frac{m-1}{2}\ \text{ copies} \\\text{of ``10''}}}1).\label{lemma3_s10}
\end{equation}
Analogously to Eqs.~\eqref{lemma3_s5} and \eqref{lemma3_s6},
Eq.~\eqref{specific_constraints_boundary} relates the coefficient of a
right-boundary domain of length $m$ to that of a bulk domain of length
$m-1$:
\begin{equation}
    (N-m+1|\underbrace{1010\cdots10}_{\substack{\frac{m-1}{2} \text{ copies} \\ \text{of ``10''}}}1)=2\lambda_{\rm II} (N-m+1|\underbrace{1010\cdots10}_{\substack{\frac{m-1}{2} \text{ copies} \\\text{ of ``10''}}}).
    \label{lemma3_s11}
\end{equation}
By the induction hypothesis, the coefficient of the single-domain state of
length $m-1$ is given by Eq.~\eqref{lemma3_2}. Therefore,
\[(N-m+1|\underbrace{1010\cdots10}_{\substack{\frac{m-1}{2} \text{ copies} \\\text{ of ``01''}}})=c_{\rm b}\lambda_{\rm II}^{m-2}.\]
Substituting this result into Eqs.~\eqref{lemma3_s10} and
\eqref{lemma3_s11}, we obtain
\begin{align}
    &(N-m+1|\underbrace{1010\cdots10}_{\substack{\frac{m-1}{2} \text{ copies} \\\text{ of ``01''}}}1)=2c_{\rm b}\lambda_{\rm II}^{m-1},\\
    &(N-m|\underbrace{1010\cdots10}_{\substack{\frac{m-1}{2}\ \text{copies} \\\text{of ``10''}}}1)=c_{\rm b}\lambda_{\rm II}^{m-1}.
\end{align}
Together with Lemma~\ref{lemma2}, these results complete the proof for $k=1$, $b_k=m-1$, and odd $m$.

For the case $b_k<m-1$, the result can be established by repeatedly using
the relations in Eq.~\eqref{specific_constraints_0}. For example, the
coefficients of single bulk-domain states with $b_k=m-1$ and $b_k=m-2$
satisfy
\begin{equation}
\begin{aligned}
    &(N-m|\underbrace{1010\cdots10}_{\substack{\frac{m-1}{2}\ \text{copies} \\\text{of ``10''}}}1)
    -(N-m|\underbrace{1010\cdots10}_{\substack{\frac{m-1}{2}\ \text{copies} \\\text{of ``10''}}}0)\\
    &+\quad\frac{\sqrt{2}}{h}(N-m|\underbrace{1010\cdots10}_{\substack{\frac{m-1}{2}\ \text{copies} \\\text{of ``10''}}})=0.
\end{aligned}
\label{lemma3_s12}
\end{equation}
From Eq.~\eqref{lemma3_s12}, we obtain
\begin{equation}
    (N-m|\underbrace{1010\cdots10}_{\substack{\frac{m-1}{2}\ \text{copies} \\\text{of ``10''}}}0)=c_{\rm b}\lambda_{\rm I}\lambda_{\rm II}^{m-2}.
\end{equation}
The corresponding result for a right-boundary domain follows from an
argument analogous to that leading to Eq.~\eqref{lemma3_s10}.
By repeatedly applying Eq.~\eqref{specific_constraints_0} together with
Lemma~\ref{lemma2}, one proves statements (1) and (2) of
Lemma~\ref{lemma3} for all single-domain states with odd $m$. The proof for
even $m$ proceeds in the same way. Moreover, since the above argument does
not involve any sites to the left of the domain under consideration, it
extends directly to configurations containing additional domains on its
left. This completes the proof of statements (1) and (2) of
Lemma~\ref{lemma3}.

It remains to verify statement (3) from the coefficient relations.
Consider an arbitrary basis state containing a single left-boundary domain,
and denote its coefficient by
$(1\mid S_1S_2\cdots S_l)$, where $l<N$. We introduce the complementary
state $\widetilde{S}_j$, defined by
$\widetilde{S}_j=1$ for $S_j=0$ and
$\widetilde{S}_j=0$ for $S_j=1$.
Consider any relation obtained from Eq.~\eqref{specific_constraints_0} that
involves $(1\mid S_1S_2\cdots S_l)$. Substituting the coefficients proposed
in statement (3) of Lemma~\ref{lemma3}, dividing the entire relation by
$c_+$ or $c_-$ according to the value of $S_1$, and multiplying it by
$c_{\rm b}$ reduces the relation exactly to the corresponding one for a
single bulk-domain state. Since statements (1) and (2) of
Lemma~\ref{lemma3} have already been proved, the resulting relation is
satisfied.
The left-boundary constraint in
Eq.~\eqref{specific_constraints_boundary} takes the form
\begin{equation}
    (1\mid S_1S_2\dots S_l)
    -(1\mid \widetilde{S}_1S_2\dots S_l)
    =\pm \frac{\xi^{S_2}_{L}}{\sqrt{2}h}
    (2\mid S_2\dots S_l),
\label{lemma3_s13}
\end{equation}
where
$\xi^{0}_{L}=2h_L-2$ and
$\xi^{1}_{L}=2h_L+2$. The positive sign applies for $S_1=1$, whereas the
negative sign applies for $S_1=0$. Substitution of statement (3) of
Lemma~\ref{lemma3} directly verifies Eq.~\eqref{lemma3_s13}. For the same
reason as in the preceding proofs, the argument extends directly to
configurations containing additional domains to the right of the
left-boundary domain.

We finally prove statement (4). Consider a basis state containing a single
domain of length $N$. Since $S_N=1$, Eq.~\eqref{specific_constraints_boundary} gives
\begin{equation}
    (1\mid S_1S_2\dots S_{N-1}1)
    =
    2\lambda_{\rm I/II}
    (1\mid S_1S_2\dots S_{N-1}),
\end{equation}
where $\lambda_{\rm I}$ applies for $S_{N-1}=1$, whereas
$\lambda_{\rm II}$ applies for $S_{N-1}=0$. Applying statement (3) of
Lemma~\ref{lemma3} to the coefficient on the right-hand side proves
statement (4).

\end{proof}

For the coefficient of an arbitrary basis state in $\ket{\Psi_T}$, repeated application of Lemma~\ref{lemma3} to the domains contained in that state directly yields Eq.~\eqref{coeff} in the main text.

\subsection{Extensions of the exact construction}\label{appB2}
The exact NESS can also be obtained within our framework for certain
modifications of the Lindbladian in Eq.~\eqref{model}, provided that HTRS is
preserved.

We first consider gain or loss acting at both boundaries with equal intensities.
The boundary longitudinal fields in this case must have the same sign when the two
dissipation are of the same type and opposite signs when they are of
different types. As an example, we take loss at the left boundary and gain
at the right boundary. The Lindbladian is
\begin{equation}
    \mathcal{L}
    =
    -i[H,\cdot]
    +\gamma\left(
    \mathcal{D}[\sigma^{-}_{1}]
    +\mathcal{D}[\sigma^{+}_{N}]
    \right),
\label{model_appB_1}
\end{equation}
where
\begin{equation}
    H
    =
    \sum^{N-1}_{j=1}\sigma^{z}_{j}\sigma^{z}_{j+1}
    +h\sum^{N}_{j=1}\sigma^{x}_{j}
    +h^{\prime}(\sigma^{z}_{1}-\sigma^{z}_{N}),
\end{equation}
and $\mathcal{D}[L]\rho=2L\rho L^{\dagger}-\left\{L^{\dagger}L,\rho\right\}$.
When the Lindbladian~\eqref{model_appB_1} satisfies HTRS, the
purification $\ket{\Psi_T}$ of its NESS is determined by
\begin{align}
    &H_{AB}\ket{\Psi_{T}}=0,
    \label{appB_hTRS_1}\\
    &(\sigma^{-}_{1,B}-\sigma^{-}_{1,A})\ket{\Psi_{T}}=0,
    \label{appB_hTRS_left}\\
    &(\sigma^{+}_{N,B}-\sigma^{+}_{N,A})\ket{\Psi_{T}}=0,
    \label{appB_hTRS_right}
\end{align}
where
\begin{equation}
H_{AB}
:=
H_A-H_B
-i\gamma\left(
\sigma^{+}_{1,A}\sigma^{-}_{1,B}
+\sigma^{-}_{N,A}\sigma^{+}_{N,B}
-\text{h.c.}
\right).
\label{appB_H_{AB}}
\end{equation}
Relative to the construction in Sec.~\ref{Sec3}, only the left-boundary
constraint is modified. The same derivation therefore applies, with
$S_1\in\{0,t\}$, $S_N\in\{1,t\}$, and
$S_i\in\{0,1,t\}$ for $1<i<N$. The coefficients retain the factorized form
of Eq.~\eqref{coeff}, with
\begin{equation}
    \lambda_{\rm I}
    =
    \frac{i\gamma-2h^{\prime}+2}{2\sqrt{2}h},
    \;
    \lambda_{\rm II}
    =
    \frac{i\gamma-2h^{\prime}-2}{2\sqrt{2}h},
    \;
    c_{\rm b}
    =
    \frac{i\gamma-2h^{\prime}}{2\sqrt{2}h}.
\label{appB_weights}
\end{equation}
The domain type factor is $c_j=c_{\rm b}$ for a bulk domain,
$c_j=2c_{\rm b}$ for a left- or right-boundary domain,
$c_j=4c_{\rm b}$ for a length-$N$ domain, and $c_j=1$ for the length-zero domain.

The above result, as well as the result in the main text,
extends straightforwardly to systems with spatially inhomogeneous couplings
and transverse fields, namely,
\begin{equation}
\begin{aligned}
    H=&
    \sum_{j=1}^{N-1}J_j\sigma_j^z\sigma_{j+1}^z
    +\sum_{j=1}^{N}h_j\sigma_j^x\\
    &+\text{boundary longitudinal fields}.    
\end{aligned}
\end{equation}
Taking Eq.~\eqref{model_appB_1} as an example, consider a basis state
$\ket{S_1S_2\cdots S_N}$ in the expansion of $\ket{\Psi_T}$ containing
$m$ domains. Let $x_k$ and $y_k$ denote the starting and
ending sites of the $k$th domain respectively. Its coefficient is then given by
\begin{equation}
    \mathrm{Coeff}(S_1S_2\cdots S_N)
    =
    \prod_{k=1}^{m}
    c_{k}^{(y_k)}
    \prod_{j=x_k}^{y_k-1}
    \lambda_{\rm I/II}^{(j)}.
\end{equation}
Here, $\lambda_{\rm I/II}^{(j)}=\lambda_{\rm I}^{(j)}$ if the bond
$(j,j+1)$ is of type I, and
$\lambda_{\rm I/II}^{(j)}=\lambda_{\rm II}^{(j)}$ if it is of type II.
The position-dependent quantities
$c_{k}^{(y_k)}$, $\lambda_{\rm I}^{(j)}$, and
$\lambda_{\rm II}^{(j)}$ are obtained directly from
Eq.~\eqref{appB_weights} by replacing $h$ and $2$ with
$h_j$ and $2J_j$, respectively.

\subsection{Integrable Structure in the Steady State}
\label{appB3}

Exact solutions of boundary-driven systems often admit MPO representations that encode an underlying Yang--Baxter structure~\cite{prosen_mps}.
It is therefore natural to ask whether the NESS obtained in the main
text exhibits a similar integrable structure. Here we use the MPO representation
\begin{equation}
    \rho_{\rm SS}
    =
    \frac{\Omega\Omega^\dagger}
    {\text{Tr}(\Omega\Omega^\dagger)} ,
\end{equation}
where $\Omega$ is obtained from $\ket{\Psi_T}$ by replacing
$\ket{0}$, $\ket{1}$, and $\ket{t}$ with
$\sigma^-$, $\sigma^+$, and $\mathbbm{1}/\sqrt{2}$,
respectively.

Since \(\text{rank}(K)=2\), the bulk tensor \(T^{[j]}\), with
\(3\leq j\leq N-2\), admits an exact reduction to a two-dimensional
auxiliary space. Introducing the additive spectral parameter
\(u=\frac{1}{2}\log({\rm i}\gamma+2h_R)\) and defining
\(x(u)=e^{2u}/(2h)\), we can write
\begin{equation}
    T^{[j]}(u)
    =
    \left[B(u)\otimes\mathbbm{1}_j\right]
    \frac{\widetilde{T}_{0j}(u)}{\sqrt{2}}
    \left[B(u)^{\rm T}\otimes\mathbbm{1}_j\right],
\label{eq:T_reduction}
\end{equation}
where \(B(u)^{\rm T}B(u)=\mathbbm{1}\), with
\begin{equation}
    B(u)=
\begin{bmatrix}
1/\sqrt{2}
&
\sqrt{\dfrac{x(u)}{\sqrt{2}[1+\sqrt{2}x(u)]}}
\\
0
&
1/\sqrt{1+\sqrt{2}x(u)}
\\
-1/\sqrt{2}
&
\sqrt{\dfrac{x(u)}{\sqrt{2}[1+\sqrt{2}x(u)]}}
\end{bmatrix}.
\label{eq:B_matrix}
\end{equation}
Here, the subscript \(0\) labels the auxiliary space, and
\begin{equation}
    \widetilde{T}_{0j}(u)
    =
    \begin{bmatrix}
        \sigma_j^x/h
        &
        {\rm i}\sqrt{x(u)/h}\,\sigma_j^y
        \\[1mm]
        {\rm i}\sqrt{x(u)/h}\,\sigma_j^y
        &
        \mathbbm{1}_j+x(u)\sigma_j^x
    \end{bmatrix}.
\label{eq:explicit_reduced_L}
\end{equation}
The reduced tensor then obeys the \(RLL\) relation
\begin{equation}
    \mathcal{R}_{0\bar{0}}(u-v)
    \widetilde{T}_{0j}(u)\widetilde{T}_{\bar{0}j}(v)
    =
    \widetilde{T}_{\bar{0}j}(v)\widetilde{T}_{0j}(u)
    \mathcal{R}_{0\bar{0}}(u-v),
\label{eq:RLL_reduced}
\end{equation}
where
\begin{equation}
    \mathcal{R}(u-v)
    =
    \begin{bmatrix}
        1&0&0&0\\
        0&\tanh(u-v)&\operatorname{sech}(u-v)&0\\
        0&\operatorname{sech}(u-v)&-\tanh(u-v)&0\\
        0&0&0&1
    \end{bmatrix}.
\label{eq:auxiliary_R}
\end{equation}
Interestingly, $\mathcal{R}(u)$ has a free-fermion six-vertex
structure, distinct from the conventional elliptic eight-vertex
structure.

\section{Steady-State Properties for $h_R<0$}
\label{appC}
We now extend the analysis of steady-state magnetization profiles and correlations in Sec.~\ref{Sec4} to \(h_R<0\).

We first focus on the generic case. In contrast to the discussion in the main text, \(h_R<0\) implies $|\lambda_{\rm I}|<|\lambda_{\rm II}|$,
so that type-II bonds are favored over type-I bonds on neighboring sites. A further distinction is that the two leading eigenvalues of the transfer matrix \(\mathbb{T}^{[j]}\), with \(3\leq j\leq N-2\), have opposite signs,
$
\operatorname{sgn}(\xi_1)\neq \operatorname{sgn}(\xi_2)
$.
It then follows directly from Eqs.~\eqref{skin_profile_estimate} and~\eqref{two_point_asymptotic} that, apart from an alternating odd-even sign structure, the steady-state magnetization profile and connected correlation function exhibit the same behavior as in the \(h_R>0\) case. This is illustrated in Fig.~\ref{fig_appC_1}.

\begin{figure}[t]
  \centering
  \includegraphics[width=0.48\textwidth]{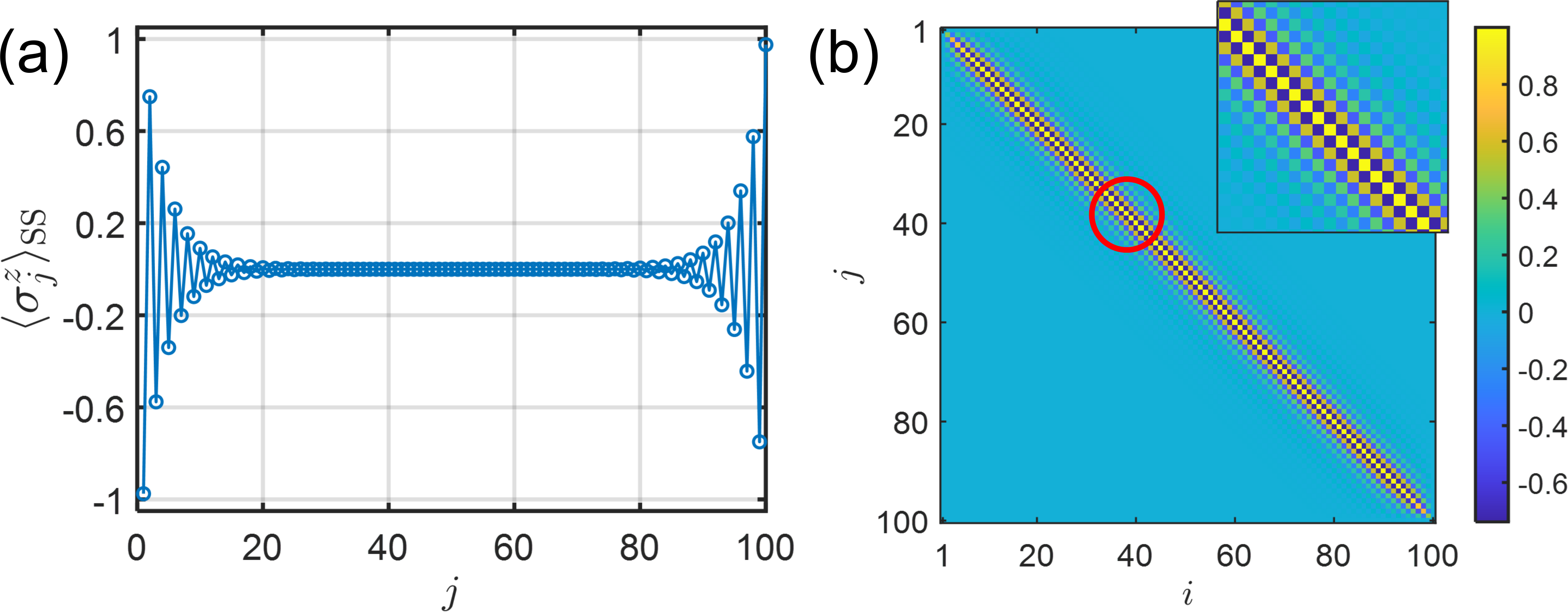}
  \caption{
   Steady-state observables in the generic case for \(h_R<0\). (a) Magnetization profile \(\langle\sigma_j^z\rangle_{\rm SS}\), exhibiting an alternating odd-even structure. (b) Connected correlation function \(C_{ij} = \langle\sigma_i^z\sigma_j^z\rangle_{\rm SS} - \langle\sigma_i^z\rangle_{\rm SS} \langle\sigma_j^z\rangle_{\rm SS}\). The inset shows an enlarged view of the region indicated by the red circle. The common parameters are \(N=100\), \(\gamma=0.1\), \(h=0.5\), and \(h_L=-h_R=2\).
  }
  \label{fig_appC_1}
\end{figure}

The analysis of the steady-state properties at $h_R=-1$ in the weak-driving limit proceeds analogously. In this case,
$|\lambda_{\rm I}|\ll |\lambda_{\rm II}|$, so the NESS overwhelmingly favors configurations dominated by type-II bonds. For $h_L\neq(-1)^{N+1}$, the thermofield state is dominated by the Néel-like configuration
$|\cdots010101\rangle$. Consequently, both the steady-state magnetization profile and the correlation functions exhibit the characteristic behavior of a Néel state, as shown in Figs.~\ref{fig_appC_2}(a) and \ref{fig_appC_2}(b).

If $h_L=(-1)^{N+1}$, the NESS is dominated by single-interface configurations, where both sides of the interface consist of type-II bonds,
\begin{equation}
\begin{aligned}
|\Psi_T\rangle &\propto
\sum_{j\,\text{ is even}}
\Bigl(
|1\cdots01\overset{j}{0}010\cdots1\rangle
+
|1\cdots101\overset{j}{t}010\cdots1\rangle
\Bigr)
\\
&+
\sum_{j\,\text{is odd}}
\Bigl(
|1\cdots010\overset{j}{1}101\cdots1\rangle
+
|1\cdots010\overset{j}{t}101\cdots1\rangle
\Bigr).
\end{aligned}
\end{equation}
Accordingly, both the magnetization profile and the correlation functions exhibit an alternating odd-even sign structure, as shown in Figs.~\ref{fig_appC_2}(c) and \ref{fig_appC_2}(d).

\begin{figure}[t]
  \centering
  \includegraphics[width=0.45\textwidth]{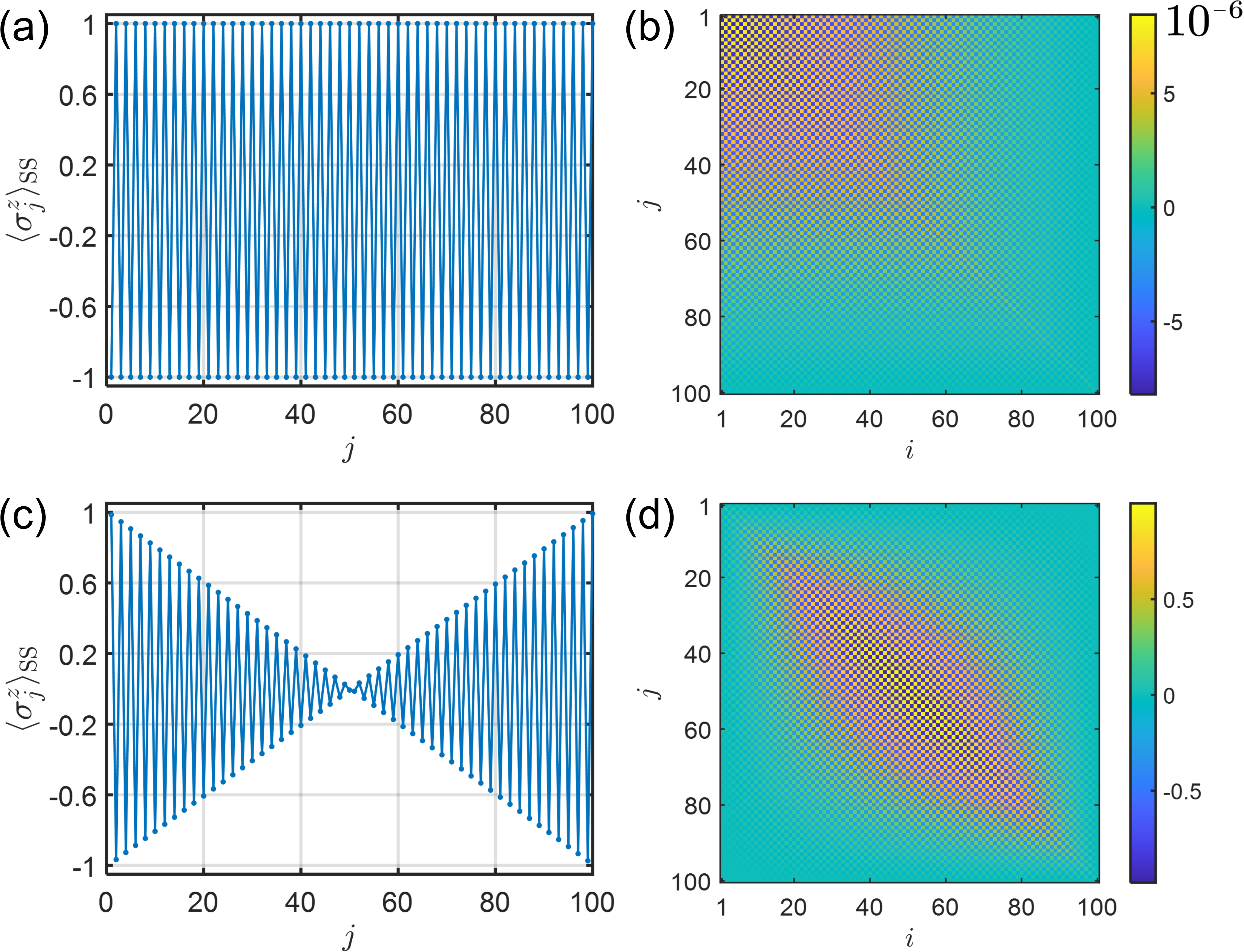}
  \caption{
Steady-state magnetization profiles and connected correlation functions in the weak-driving limit for $h_R=-1$.
(a) Magnetization profile for $h_L=2$, showing the Néel-like ordering.
(b) Connected correlation function for $h_L=2$. The small deviation from the ideal Néel-state behavior originates from subleading configurations induced by finite values of $h$ and $\gamma$.
(c) Magnetization profile for $h_L=-1$.
(d) Connected correlation function for $h_L=-1$. The common parameters are $N=100$ and $h=\gamma=0.001$.
  }
  \label{fig_appC_2}
\end{figure}

\section{Corrections in the Weak Driving Limit}
\label{appD}

The pictures discussed in Sec.~\ref{Sec4} become exact only in the limit of
infinitesimal $h$ and $\gamma$. At finite $h$ and $\gamma$, configurations
in $|\Psi_T\rangle$ containing additional interfaces acquire nonzero
amplitudes and modify the ideal behavior.

To quantify this effect, we denote by $W_n$ the total weight of the
$n$-interface sector in the nonequilibrium partition function
$\mathcal{Z}=\langle\Psi_T|\Psi_T\rangle$. Up to boundary-dependent
prefactors, it can be estimated as
\begin{equation}
    W_n
    \sim
    \mathcal{N}_n \varepsilon^n,
    \qquad
    \mathcal{N}_n\sim\binom{N-1}{n},
    \label{eq:Wn_scaling}
\end{equation}
where $\mathcal{N}_n$ counts the possible interface positions and
$\varepsilon$ denotes the characteristic relative weight associated with
one interface. Taking $h_R=1$ as an example, an interface can be represented
either by a type-II bond or by a local $\ket{t}$ state connecting two
type-I domains. Its characteristic relative weight can therefore be
estimated as
\begin{equation}
    \varepsilon
    \simeq
    \left|
    \frac{\lambda_{\rm II}}{\lambda_{\rm I}}
    \right|^2
    +
    \left|
    \frac{c_{\rm b}}{\lambda_{\rm I}^2}
    \right|^2
    \simeq
    \frac{\gamma^2}{16}
    +
    \frac{h^2}{8},
    \label{eq:interface_weight}
\end{equation}
in the weak-driving limit. Thus, although the weight of an individual multi-interface
configuration is exponentially suppressed with the number of interfaces,
the total number of such configurations grows combinatorially with $N$.

Specifically, for $h_L\neq-1$, the fully polarized configuration is
dominant. Once $h$ and $\gamma$ are finite, however, the subleading
single-interface sector acquires a nonzero weight. It produces a slight
deviation of the magnetization profile from perfect polarization and
generates weak long-range correlations. These corrections become
increasingly visible as $N$ grows.

For $h_L=-1$, the single-interface sector is dominant, and the leading
correction is provided by the two-interface sector. Their total-weight ratio
can be estimated as
\begin{equation}
    \frac{W_2}{W_1}
    \simeq
    \frac{\mathcal{N}_2}{\mathcal{N}_1}
    \varepsilon
    \simeq
    \frac{N\varepsilon}{2}.
\end{equation}
If the correction is regarded as appreciable when
$W_2/W_1\simeq \delta$, the corresponding crossover size is
\begin{equation}
    N_{c}
    \simeq
    \frac{2\delta}
    {\gamma^2/16+h^2/8}.
    \label{eq:Nc}
\end{equation}
For $\delta=0.1$ and the parameters used in the main text,
$h=\gamma=0.001$, this estimate gives $N_{c}\simeq 1\times10^6$.
Thus, for $h_R=\pm 1$ and $h_L=-1$, the conclusions drawn in the main text remain valid
either for system sizes below the crossover scale, $N<N_{\rm c}$, or,
equivalently, over a broad range of small but finite $h$ and $\gamma$ at
fixed $N$.

Overall, within the weak-driving regime, any finite values of $h$ and
$\gamma$ render multi-interface configurations increasingly important as
$N$ grows. In the limit $N\rightarrow\infty$, the NESS gradually recovers
the generic case behavior, unless $\gamma$ and $h$ are required to be infinitesimally small. This mechanism parallels that of a classical Ising chain with a longitudinal
field: the single-kink configuration dominates at zero temperature, whereas
finite temperature gives exponentially small weights to multi-kink
configurations whose growing number eventually destroys this picture as
$N$ increases. Finite $h$ and $\gamma$ here play a role similar to finite temperature.

\bibliographystyle{apsrev4-2-title}
\bibliography{refs_prb_titles.bib}

\end{document}